\documentclass[11pt,letterpaper]{article}

\usepackage[utf8]{inputenc}
\usepackage[english]{babel}
\usepackage[margin=1in]{geometry}

\usepackage{amsmath}
\usepackage{amsthm, amsfonts}
\usepackage{mathtools, thmtools}
\usepackage{bbm}

\newtheorem{definition}{Definition}[section]

\newtheorem{theorem}{Theorem}[section]
\newtheorem{lemma}{Lemma}[section]
\newtheorem{claim}{Claim}[section]

\usepackage{subcaption}
\usepackage{float}
\usepackage{epsfig}
\usepackage{tikz}

\usepackage{paralist}
\usepackage{enumerate}
\usepackage{cases}
\usepackage{caption}
\usepackage{multicol}
\usepackage{graphicx}
\usepackage{xcolor}
\usepackage{xspace}
\usepackage{complexity}

\usepackage{ifthen}
\usepackage{algorithm}
\usepackage{algorithmic}

\usepackage[numbers,sort]{natbib}

\usepackage{url}

\usepackage{hyperref}

\newcommand{\marginal}{\lambda}
\newcommand{\yminus}[1]{\vec{y}_{\text{-}#1}}
\newcommand{\Gminus}[1]{G_{\text{-}#1}}

\newcommand{\fomp}{\text{fully online matching}\xspace}

\newcommand{\ranking}{Ranking\xspace}
\newcommand{\wtf}{Water-filling\xspace}
\newcommand{\ewf}{Eager Water-filling\xspace}
\newcommand{\br}{Balanced Ranking\xspace}

\newcommand{\eqdef}{\stackrel{\textnormal{def}}{=}}

\newcommand{\vect}[1]{\ensuremath{\vec{#1}}}
\newcommand{\vecy}{\vect{y}}
\newcommand{\vecmv}[1][v]{\vect{y}_{\text{-}#1}}

\newcommand{\pw}{p}
\newcommand{\dd}[1]{\mathrm{d}#1}
\renewcommand{\E}{\mathbb{E}}
\newcommand*\samethanks[1][\value{footnote}]{\footnotemark[#1]}

\title{Fully Online Matching II: \\ Beating Ranking and Water-filling}
\author{Zhiyi Huang\thanks{The Univerisity of Hong Kong. \texttt{\{zhiyi, yhzhang2\}@cs.hku.hk}} 
    \and Zhihao Gavin Tang\thanks{ITCS, Shanghai University of Finance and Economics. \texttt{tang.zhihao@mail.shufe.edu.cn}}
    \and Xiaowei Wu\thanks{IOTSC, University of Macau. \texttt{xiaoweiwu@um.edu.mo}}
    \and Yuhao Zhang\samethanks[1]}
\date{}

\begin{document}

\begin{titlepage}
\thispagestyle{empty}
\maketitle
\begin{abstract}
    \thispagestyle{empty}
    Karp, Vazirani, and Vazirani (STOC 1990) initiated the study of online bipartite matching, which has held a central role in online algorithms ever since.
    Of particular importance are the Ranking algorithm for integral matching and the Water-filling algorithm for fractional matching.
    Most algorithms in the literature can be viewed as adaptations of these two in the corresponding models.
    Recently, Huang et al.~(STOC 2018, SODA 2019) introduced a more general model called \emph{fully online matching}, which considers general graphs and allows all vertices to arrive online.
    They also generalized Ranking and Water-filling to fully online matching and gave some tight analysis:
    Ranking is $\Omega \approx 0.567$-competitive on bipartite graphs where the $\Omega$-constant satisfies $\Omega e^\Omega = 1$, and Water-filling is $2-\sqrt{2} \approx 0.585$-competitive on general graphs.
    
    We propose fully online matching algorithms strictly better than Ranking and Water-filling.
    For integral matching on bipartite graphs, we build on the online primal dual analysis of Ranking and Water-filling to design a $0.569$-competitive hybrid algorithm called Balanced Ranking.
    To our knowledge, it is the first integral algorithm in the online matching literature that successfully integrates ideas from Water-filling.
    For fractional matching on general graphs, we give a $0.592$-competitive algorithm called Eager Water-filling, which may match a vertex on its arrival.
    By contrast, the original Water-filling algorithm always matches vertices at their deadlines.
    Our result for fractional matching further shows a separation between fully online matching and the general vertex arrival model by Wang and Wong (ICALP 2015), due to an upper bound of $0.5914$ in the latter model by Buchbinder, Segev, and Tkach (ESA 2017).
\end{abstract}

\end{titlepage}

\section{Introduction}
\label{sec:introduction}

Online matching is one of the oldest and most fruitful topic in the online algorithms literature.
It dates back to thirty years ago when \citet{stoc/KarpVV90} proposed the online bipartite matching problem and the Ranking algorithm.
Consider a bipartite graph, where the left-hand-side vertices are offline, i.e., known upfront to the algorithm, and the right-hand-side vertices are online arriving one at a time.
On the arrival of an online vertex, the algorithm observes its incident edges and must immediately and irrevocably decide how to match it.
The goal is to maximize the cardinality of the matching.
For a real-world example, think of online advertising where the offline and online vertices correspond to advertisers and impressions respectively.
Ranking picks a random permutation of the offline vertices, and matches each online vertex to the first unmatched neighbor by the permutation.
\citet{stoc/KarpVV90} showed that it is $1-\frac{1}{e}$-competitive, and this is the best possible for the problem.

In online advertising, an advertiser can usually be matched to many impressions.
This is the $b$-matching model of \citet{tcs/KalyanasundaramP00} where $b$ is the number of times an offline vertex can be matched, a.k.a., its capacity.
This is also closely related to the fractional relaxation of online bipartite matching where each online vertex may be matched fractionally to multiple offline neighbors so long as the total matched amount does not exceed one unit.%
\footnote{The fractional problem is equivalent to a $b$-matching problem in which online vertices arrive in batches of $b$ copies and $b$ tends to infinity.
Further, the assumption of having $b$ copies per online vertex is irrelevant in existing analysis.}
In this case, the optimal $1-\frac{1}{e}$ competitive can be achieved with a deterministic algorithm called Water-filling (a.k.a.\ Water-level or Balance).
It matches each online vertex continuously to the least matched offline neighbor.

Ranking and Water-filling are the most fundamental algorithms in online bipartite matching and its variants.
Most algorithms in the online matching literature under worst-case analysis can be viewed as adaptations of them in the corresponding models.

\paragraph{Fully Online Matching.}
Let us turn to a different real-world scenario which involves a bipartite matching problem with an online flavor.
Consider a ride-hailing platform that matches drivers on one side and passengers on the other side.
This is \emph{not} captured by the model of \citet{stoc/KarpVV90} because vertices on both sides of the bipartite graph arrive and depart online.
To this end, \citet{stoc/HuangKTWZZ18, soda/HuangPTTWZ19} recently introduced a generalized model called \emph{fully online matching}.
Each vertex in the fully online model is associated with not only an arrival time but also a deadline.
Each edge is revealed to the algorithm when both endpoints have arrived, and can be selected into the matching anytime before the endpoints' deadlines provided that they are still unmatched.
Furthermore, the model extends naturally to general graphs, capturing an even broader class of problems including matching passengers in ride-sharing.

\citet{stoc/HuangKTWZZ18, soda/HuangPTTWZ19} generalized both Ranking and Water-filling to fully online matching.
Both algorithms only match vertices at their deadlines.
Ranking selects a random permutation of \emph{all vertices};%
\footnote{For example, draw a random number in $[0, 1)$ on the arrival of each vertex and sort them by the numbers.}
then, for any vertex that stays unmatched till its deadline, Ranking matches it to the first available neighbor by the permutation.
Similarly, for any vertex that is not fully matched by its deadline, Water-filling matches the remaining portion fractionally to the least-matched available neighbors.
They showed that Ranking is $0.521$-competitive on general graphs.
For bipartite graphs, they gave an tight analysis that Ranking is $\Omega \approx 0.567$-competitive, where the $\Omega$-constant is the solution of $\Omega \cdot e^\Omega = 1$.
Further, they proved a tight $2-\sqrt{2} \approx 0.585$ competitive ratio of Water-filling on general graphs.
Finally, they separated fully online matching with online bipartite matching of \citet{stoc/KarpVV90} by showing that there is no $1-\frac{1}{e}$-competitive algorithm in the fully online model.

\subsection{Our Contributions}

This work is driven by a natural question:
\emph{Are Ranking and Water-filling optimal in fully online matching, like in many other online matching models?}
In particular, is the $\Omega \approx 0.567$ competitive ratio the best possible for integral algorithms on bipartite graphs?
How about the $2-\sqrt{2} \approx 0.585$ competitive ratio for fractional algorithms on general graphs?
Surprisingly, the answers are no!
There are algorithms strictly better than Ranking and Water-filling in fully online matching!

\paragraph{Beating Ranking on Bipartite Graphs.}
We follow a simple intuition:
since Water-filling has a superior competitive ratio, we may ``correct'' the decisions by Ranking with those by Water-filling.
While easy to state, this intuition is difficult to substantiate.
In fact, to our knowledge, there is no integral algorithm in the online matching literature prior to our work which successful integrates ideas from Water-filling.
To explain our algorithm, we need the following equivalent interpretations of Ranking and Water-filling from the online primal dual technique (see, e.g., \citet{soda/DevanurJK13}).

\begin{itemize}
    \item \textbf{Ranking:~}
        Draw a random number $y_u \in [0, 1]$ for each vertex $u$.
        Then, if $v$ is matched to vertex $u$ at $u$'s deadline, they split one unit of gain.
        Vertex $v$ keeps $g(y_v) = e^{y_v-1}$ to itself, and offers $1 - g(y_v)$ to $u$.
        For each vertex $u$ which stays unmatched till its deadline, it matches to the offline neighbor who offers the most, i.e., the one with the smallest $y_v$.
    \item \textbf{Water-filling:~}
        For each vertex $u$, let $x_u \in [0, 1]$ denote its matched portion, a.k.a.\ its water level.
        Then, if $v$ is matched to vertex $u$ at $u$'s deadline by some infinitesimal amount $\epsilon$, they split the gain of $\epsilon$.
        Vertex $v$ keeps $f(x_v) = e^{x_v-1}$ times $\epsilon$ to itself, and offers $\big(1 - f(x_v)\big) \epsilon$ to $u$.
        For each vertex $u$ which is not fully matched by its deadline, it fractionally matches to the least matched offline neighbors to maximize the total offer.
\end{itemize}

For any nondecreasing $f$ and $g$, we define a hybrid algorithm called Balance Ranking as follows.

\begin{itemize}
    \item \textbf{Balanced Ranking:~}
        For each vertex $u$, draw a random number $y_u \in [0, 1]$.
        Let $x_u$ be the probability that $u$ is matched, a.k.a.\ its water level.
        Then, if $v$ is matched to another vertex $u$ at $u$'s deadline, they split one unit of gain.
        Vertex $v$ keeps $f(x_v) + g(y_v)$ to itself, and offers $1 - f(x_v) - g(y_v)$ to $u$, \emph{where $x_v$ is $v$'s water level after $u$'s deadline}.
        For each vertex $u$ which stays unmatched till its deadline, it matches to the offline neighbor who offers the most.
\end{itemize}

While the algorithm is a simple combination the alternative interpretations of Ranking and Water-filling, it is crucial to match each vertex based on the water levels of the neighbors \emph{after} the matching decision of the current vertex.
We show in Section~\ref{subsec:predicting} that it is a well-defined algorithm.

Then, we analyze Balanced Ranking under the online primal dual framework and design the functions $f$ and $g$ by solving a differential equation arose from the analysis.
See Section~\ref{subsec:analysis-balanced-ranking}.

\begin{theorem}
    \label{thm:balanced-ranking}
    Balanced Ranking is $0.569$-competitive for fully online matching on bipartite graphs.
\end{theorem}

\paragraph{Beating Water-filling.}
We start with an observation that Water-filling works in an even harder model, where an edge is revealed to the algorithm only when an endpoint reaches the deadline.
In fact, the hardness result by \citet{soda/HuangPTTWZ19} implies that Water-filling is optimal in the harder model.
The observation suggests, however, Water-filling gives up the information about the edges among the vertices which have arrived but have not yet reached the deadlines.
Intuitively, we shall be able to improve the competitive ratio by taking such information into account.
Our algorithm utilizes the information implicitly by eagerly matching vertices partially on the arrivals.
Indeed, the eager matches are precisely among vertices that have arrived but have not reached the deadlines.
For any nondecreasing function $f$, define the Eager Water-filling algorithm as follows.

\begin{itemize}
    \item \textbf{Eager Water-filling:~}
        On the arrival of each vertex $u$, match it fractionally to the least matched offline neighbors $v$ as long as $v$'s offer is larger than what $u$ wants for itself at the current water level, i.e., $1-f(x_v) \ge f(x_u)$.
        If a vertex $u$ is not fully matched by its deadline, match it fractionally to the least matched offline neighbors to maximize the total offer.
\end{itemize}

We can naturally interpret the algorithm as having vertex $u$ make decisions \emph{assuming the graph stays as it is}.
Suppose some neighbor $v$ satisfies $1-f(x_v) \ge f(x_u)$ on $u$'s arrival.
On the one hand, an eager match with $v$ offers $1-f(x_v)$ to $u$ per unit of match.
On the other hand, if $u$ opts to wait, it risks getting matched at some other vertex's deadline in which case $u$ keeps only $f(x_u)$ per unit of match, inferior to an eager match with $v$.
Further, $v$ may be matched by some other vertex while $u$ waits. 
Finally, even if none of these happens, $u$ at best has the same options at its deadline compared to the eager matches on its arrival, assuming the graph stays the same.
In sum, $u$ shall fractionally match to $v$ on its arrival as in Eager Water-filling.

Again, we analyze Eager Water-filling under the online primal dual framework and design the function $f$ by solving a differential equation arose from the analysis.
See Section~\ref{sec:ewf}.

\begin{theorem}
    \label{thm:eager-wf}
    Eager Water-filling is $0.592$-competitive for fractional fully online matching.
\end{theorem}

This result separates fully online matching with another model called general vertex arrival by \citet{icalp/WangW15} because of a $0.5914$ upper bound on the best possible competitive ratio in the latter model by \citet{esa/BuchbinderST17}.
Intriguingly, general vertex arrival is essentially fully online matching \emph{restricted to eager matches only}.
In other words, we obtain the separation by forfeiting part of the flexibility to defer decisions till the deadlines, and by incorporating eager matches which are allowed in the general vertex arrival model in the first place.

\subsection{Other Related Works}

The analysis of Ranking in online bipartite matching has been refined and simplified in a series of papers by \citet{soda/GoelM08}, \citet{sigact/BenjaminM08}, and \citet{soda/DevanurJK13}.
In particular, the online primal dual framework by \citet{soda/DevanurJK13} has been the backbone of the competitive analysis in fully online matching including those in this paper.

Many variants of online bipartite matching have been introduced.
\citet{focs/MehtaSVV05} proposed the first generalization called AdWords motivated by online advertising, which has been simplified and generalized under online primal dual~\cite{esa/BuchbinderJN07, stoc/DevanurJ12}.
\citet{soda/AggarwalGKM11} considered the vertex-weighted problem and extended Ranking to this model.
\citet{wine/FeldmanKMMP09} investigated the fractional edge-weighted case.
Their algorithm can be seen as an adaptation of Water-filling, and the analysis was simplified by \citet{teac/DevanurHKMY16}.
\citet{focs/MehtaP12} introduced a model with stochastic rewards and the results were later improved by \citet{soda/MehtaWZ14} and \citet{stoc/HuangZ20}.
Some of the models have also been studied under the assumption of a random arrival order~\cite{stoc/KarandeMT11, stoc/MahdianY11, talg/HuangTWZ19}.

Besides fully online matching and general vertex arrival, there is an even harder edge arrival model.
Recently, \citet{focs/GamlathKMSW19} proved that the trivial $0.5$-competitive greedy algorithm is the best possible.
They also obtained the first integral algorithm that breaks the $0.5$ barrier in the general vertex arrival model.
Prior to that, there were some positive results for special cases of edge arrival, e.g., when the graph is a forest~\cite{esa/BuchbinderST17}.
The fully online matching problem is closely related to the online windowed matching problem by \citet{ec/AshlagiBDJSS19}, which can be viewed as an edge-weighted version of fully online matching under the first-in-first-out assumption.

\section{Preliminaries}
\label{sec:preliminaries}

\paragraph{Model.}
Consider an undirected graph $G = (V, E)$.
Initially, the algorithm has no information about $G$.
Then, we proceed in $2|V|$ steps, each of which is one of the following two kinds:
\begin{itemize}
    \item \emph{Arrival of a vertex $u$:~}
        The algorithm observes the edges between $u$ and the previously arrived vertices.
        This is the earliest step when $u$ can be matched.
    \item \emph{Deadline of a vertex $u$:~}
        This is the last step when $u$ can be matched.
        We guarantee that all neighbors of $u$ arrive before $u$'s deadline.%
        \footnote{Consider the ride-hailing example.
        The guarantee effectively means that, for instance, a driver on a day shift cannot be matched with a passenger in the evening.}
\end{itemize}

The goal is to maximize the size of the matching.
Following the standard competitive analysis, an algorithm is $\Gamma$-competitive for some $0 \le \Gamma \le 1$, if for any fully online matching instance, the expected size of its matching is at least $\Gamma$ times the optimal matching in hindsight.

Observe that fully online matching generalizes the model of \citet{stoc/KarpVV90}, because the latter can be seen as having the offline vertices arrive at the beginning and leave at the end, and letting the deadline of each online vertex be right after its arrival.

\paragraph{Integral vs.\ Fractional Algorithms.}
An integral algorithm must match each vertex $u$ in whole to another vertex, although the matching decisions could be randomized.
A fractional algorithm, however, may match a vertex $u$ fractionally to multiple vertices, e.g., $\frac{1}{2}$ to $v_1$, $\frac{1}{4}$ to $v_2$, and another $\frac{1}{4}$ to $v_3$, as long as the total amount is at most $1$.

\paragraph{Matching LP.}
For any edge $(u, v)\in E$, let $x_{uv}$ be the probability/fraction that edge $(u, v)$ is matched by the algorithm.
Consider the following standard matching LP and its dual:
\begin{align*}
\max: \quad & \textstyle \sum_{(u,v)\in E} x_{uv} && \qquad\qquad & \min: \quad & \textstyle\sum_{u \in V} \alpha_u\\
\text{s.t.} \quad & \textstyle \sum_{v:(u,v)\in E} x_{uv} \leq 1 && \forall u\in V & \text{s.t.} \quad & \alpha_u + \alpha_v \geq 1 && \forall (u,v)\in E \\
& x_{uv} \geq 0 && \forall (u,v)\in E & & \alpha_u \geq 0 && \forall u \in V
\end{align*}

Let $P$ and $D$ denote the primal and dual objectives respectively.
Observe that by the above choice of $x_{uv}$'s, $P$ also equals the expected size of the algorithm's matching.

\paragraph{Randomized Online Primal Dual Framework.}
An online primal dual algorithm maintains not only a matching but also a dual assignment online.

\begin{lemma}[\citet{soda/DevanurJK13}]
    \label{lem:primal-dual}
    An online primal dual algorithm is $\Gamma$-competitive if we have:
    \begin{itemize}
        \item \emph{Approximate dual feasibility in expectation:~}
            $\forall (u, v) \in E,~\E \big[ \alpha_u \big] + \E \big[ \alpha_v \big] \ge \Gamma$;
        \item \emph{Reverse weak duality in expectation:~}
            $P \ge \E \big[ D \big]$.
    \end{itemize}
\end{lemma}

The algorithms in this paper will satisfy reverse weak duality in expectation with equality.
This is because whenever an edge $(u, v)$ is matched by our algorithms, the increment in matching size is split between the dual variables $\alpha_u$ and $\alpha_v$ of the two endpoints.

\section{Balanced Ranking} \label{sec:balanced-ranking}

This section presents the \br algorithm, which is a hybrid algorithm building on both \ranking and \wtf for \fomp on bipartite graphs, and prove Theorem~\ref{thm:balanced-ranking}.

\subsection{Matching with Ranks and Lookahead Water Levels}
\label{subsec:predicting}

Recall the primal dual interpretations of \ranking and \wtf as follows.
\ranking fixes a nondecreasing function $g$, and draws a random rank $y_u \in [0, 1]$ for each vertex $u$. At the deadline of each vertex $u$, if $u$ is not matched yet the algorithm matches it to its neighbor $v$ with the largest offer $1 - g(y_v)$.
\wtf maintains the matched fraction $x_u \in [0, 1]$ of each vertex $u$, a.k.a.\ its water level, and matches $u$ fractionally to the neighbors with the largest offers $1 - f(x_v)$ per unit of match for some nondecreasing function $f$.
Hence, a natural hybrid algorithm is to define the offer of each vertex $v$ to be $1 - g(y_v) - f(x_v)$, and to match the neighbor with the largest offer.
Here, the water level $x_v$ in a randomized integral algorithm is the probability that $v$ is matched;
this is equivalent to the probability that it is \emph{passive} since any relevant $v$ has not reached its deadline.

\paragraph{Lookahead Water Levels.}
Observe, however, the water levels change over time.
Therefore, we need to further elaborate at what time we evaluate the water levels $x_v$'s in the hybrid algorithm.
Suppose we are to match a vertex $u$ which stays unmatched by its deadline.
The first instinct may be to use the current water levels right \emph{before} the deadline of $u$.
Surprisingly, the attempt fails according to our analysis.
We instead consider the water levels right \emph{after} the deadline of $u$, which we call the \emph{lookahead water levels}.
Intuitively, balancing the lookahead water levels keeps as many options available as possible to hedge against all future possibilities.

To avoid confusion, we use $x_v^{(u)}$ to denote $v$'s water level right after $u$'s deadline. It exactly equals to the probability that $v$ is passive after $u$'s deadline. 




\paragraph{Computing Lookahead Water Levels.}
The algorithm is still incomplete as it involves circular definitions.
The lookahead water levels $x_v^{(u)}$'s depend on the matching decision at $u$'s deadline, which is made based on the lookahead water levels $x_v^{(u)}$'s.
Next we argue this is not only well defined but further efficiently computable up to high accuracy. Assuming:
\begin{itemize}
    \item $f$ is $1$-Lipschitz i.e., $f(x) - f(y) \le x-y$ for any $x \ge y \in [0, 1]$;
    \item $g$ is $\frac{1}{100}$-reverse Lipschitz, i.e., $g(x) - g(y) \ge \frac{x-y}{100}$ for any $x \ge y \in [0, 1]$.
\end{itemize}

The constants $1$ and $\frac{1}{100}$ are arbitrary so long as the former is not too small and the latter is not too large.
They are only for the convenience in the definition of the algorithm and are not binding constraints in our analysis.
Even if not stated explicitly, there is some optimal choice of $f$ and $g$ in our competitive analysis with the above Lipschitz and reverse Lipschitz properties.

\begin{lemma} \label{lem:lookahead-water level}
    Suppose the algorithm is well defined before $u$'s deadline.
    In $\poly(|V|, \frac{1}{\epsilon})$ time, we can compute $\widehat{f}_v \in [0, 1]$ for all $v \in V$, such that whenever $u$ stays unmatched by its deadline, matching $u$ to vertex $v$ with the largest $1 - g(y_v) - \widehat{f}_v$ leads to water levels $x_v^{(u)}$'s with $\widehat{f}_v - \epsilon \le f(x_v^{(u)}) \le \widehat{f}_v$.
\end{lemma}
\begin{proof}
    We start with the trivial overestimates $\widehat{f}_v = 1$ for all $v \in V$.
    Then, we iteratively refine them while keeping the invariant that they are overestimates.
    That is, whenever $u$ stays unmatched by its deadline, matching $u$ to the vertex $v$ with the largest $1 - g(y_v) - \widehat{f}_v$ leads to water levels $x_v^{(u)}$'s such that $f(x_v^{(u)}) \le \widehat{f}_v$.
    Finally, we bound the time complexity by proving that the sum of the estimates, i.e., $\sum_{v \in V} \widehat{f}_v$, decreases at least linearly.

    Concretely, whenever there is a vertex $v$ with $\widehat{f}_v - f(x_v^{(u)}) > \epsilon$, decrease the estimate $\widehat{f}_v$ by $\frac{\epsilon}{101}$.
    Here, we can compute $x_v^{(u)}$ up to high enough accuracy from sample runs of the algorithm by standard concentration bounds.
    In doing so, the water level $x_v^{(u)}$ increases and the water levels $x_w^{(u)}$ for all $w \notin \{u, v\}$ weakly decreases.
    
    We first argue that the invariant still holds.
    It suffices to consider $v$ because for any other vertex the water level weakly decreases and the estimate stays the same.
    Next we show that the water level of $v$ increases by at most $\frac{100 \epsilon}{101}$ and thus, maintains the invariant.
    Equivalently, we claim that the probability $v$ is matched at $u$'s deadline increases by at most $\frac{100 \epsilon}{101}$, which is true even conditioned on the ranks $\yminus{v}$ of the other vertices.
    By decreasing $\widehat{f}_v$ by $\frac{\epsilon}{101}$, the threshold $g(y_v)$ above which $v$ is picked by $u$, increases by the same amount.
    This in turn increases the threshold rank $y_v$ by at most $\frac{100\epsilon}{101}$ because $g$ is $\frac{1}{100}$-reverse Lipschitz.
    Since $y_v$ is uniform from $[0, 1]$, we conclude that the probability that $v$ is matched at $u$'s deadline increases by at most $\frac{100\epsilon}{101}$, conditioned on \emph{any} $\yminus{v}$.
    
    Finally, the algorithm terminates in $O(\frac{|V|}{\epsilon})$ iterations, because the sum of the estimates, i.e., $\sum_{v \in V} \widehat{f}_v$, decreases by at at least $\Omega(\epsilon)$ per iteration and it is between $0$ and $|V|$.
\end{proof}

Our analysis degrades gracefully in the error term $\epsilon$ in the above lemma.
For simplicity, the rest of the section assumes the limit case when $f(x_v^{(u)}) = \widehat{f}_v$.
See Algorithm~\ref{alg:balanced_ranking}.



\begin{algorithm}[t]
    \caption{\br, with Dual Assignments}
    \label{alg:balanced_ranking}
    \begin{algorithmic}
        \STATE \textbf{at $u$'s arrival:}
        \STATE\hspace{\algorithmicindent}
            draw rank $y_u \in [0, 1]$ uniformly at random\\[1ex]
        \STATE \textbf{at $u$'s deadline, if it is unmatched:}
        \STATE\hspace{\algorithmicindent}
            compute the lookahead water levels $x_v^{(u)}$ for all $v \in V$ \hspace*{\fill} (Lemma~\ref{lem:lookahead-water level})
        \STATE\hspace{\algorithmicindent}
            match $u$ to the unmatched neighbor $v$ with the largest $1 - g(y_v) - f(x_v^{(u)})$
        \STATE\hspace{\algorithmicindent}
            let $\alpha_u = 1 - g(y_v) - f(x_v^{(u)})$ and $\alpha_v = g(y_v) + f(x_v^{(u)})$
    \end{algorithmic}
\end{algorithm}

\subsection{Notations and Basic Properties} \label{sec:br_property}

Our analysis of \br builds on the approach of \citet{stoc/HuangKTWZZ18, soda/HuangPTTWZ19}.
This section adopts some notations from their analysis, and establishes several basic properties of \ranking that continue to hold for \br.


In the following, for any instance $G$ and any realization of ranks $\vec{y}$, let $M_G(\vec{y})$ be the matching produced by \br.
Let $\Gminus{u}$ be the instance with $u$ removed from $G$.
If $G$ is clear in the context, we omit the subscript to write $M_G(\vec{y})$ as $M(\vec{y})$, and $M_{\Gminus{u}}(\yminus{u})$ as $M(\yminus{u})$.
We remark that when running \br on instance $\Gminus{u}$, the lookahead water levels remain defined by instance $G$.
In other words, in the thought experiment that removes $u$, we assume that the ranks and lookahead water levels of vertices other than $u$ remain unchanged.

\begin{definition}[Active and Passive]
    If an edge $(u, v)$ is matched in $M(\vec{y})$ at $u$'s deadline, we say that $u$ is \emph{active} and $v$ is \emph{passive}.
\end{definition}

The roles of active and passive vertices in the analysis are similar to the online and offline vertices respectively in the model of \citet{stoc/KarpVV90}.

A main structural property of \ranking is the alternating path property that characterizes how the matching changes when the rank of a vertex changes.
It also holds to \br.

\begin{lemma}[Alternating Path] \label{lem:alternating-path}
	In a bipartite instance $G$, if $u$ is matched in $M(\vec{y})$, no neighbor of $u$ gets better from $M(\vec{y})$ to $M(\yminus{u})$.
	Here, passive is better than active, and active is better than unmatched.
	Conditioned on being passive, it is better to match a vertex with an earlier deadline.
	Conditioned on being active, it is better to match a vertex $v$ with larger $1- g(y_v)-f(x_v)$.
\end{lemma}
\begin{proof}
	Recall that at the deadline of a vertex $u$, it chooses the available neighbor with the largest $1-g(y_v)-f(x_v^{(u)})$ by Balanced Ranking. Within the proof, we only use the property that $f(x_v^{(u)})$ is a globally fixed quantity that does not depends on the realization of the ranks $\vec{y}$. In order words, Balanced Ranking has the property that at the deadline of any vertex $w$, if $z_1$ has higher priority than $z_2$ in $G$, i.e. $1-g(y_{z_1})-f(x_{z_1}^{(w)}) > 1-g(y_{z_2})-f(x_{z_2}^{(w)})$, then $z_1$ remains having higher priority than $z_2$ in $G_{-u}$. This is the crucial property of Ranking for Lemma 2.5 of \cite{stoc/HuangKTWZZ18} to hold. The remaining of the proof is almost verbatim to that of \cite{stoc/HuangKTWZZ18}.
	
	We prove that the symmetric difference between the matchings $M(\vec{y})$ and $M(\vecmv[u])$ is an alternating path $(u_0=u,u_1,\cdots,u_l)$ such that
	\begin{enumerate}
		\item for all even $i<l$, $(u_i,u_{i+1}) \in M(\vec{y})$; for all odd $i<l$, $(u_i,u_{i+1}) \in M(\vecmv[u])$;
		\item from $M(\vec{y})$ to $M(\vecmv[u])$, vertices $\{u_1,u_3,\cdots\}$ get worse, vertices $\{u_2,u_4,\cdots\}$ get better.
	\end{enumerate}
	We prove the statement by mathematical induction on $n$, the total number of vertices. For the base case when $n=2$, the symmetric difference is a single edge $(u,u_1)$ and the second statement holds since $u_!$ is matched in $M(\vec{y})$ but unmatched with $u$ removed.
	
	Suppose the lemma holds for $1,2,\cdots,n-1$. We consider the case when there are $n$ vertices. Let $u_1$ be matched to $u$ in $M(\vec{y})$. If we remove both $u,u_1$ from $G$ (let $\vec{y'}=[0,1]^{V\backslash \{u,u_1\}}$ be the resulting vector), then we have $M(\vec{y})=M(\vec{y'}) \cup \{(u,u_1)\}$.
	
	If $u_1$ is unmatched in $M(\vecmv[u])$, we have $M(\vecmv[u]) = M(\vec{y'})$ and the lemma holds. Now suppose $u_1$ is matched in $M(\vecmv[u])$.
	
	By definition $\vect{y'}$ is obtained by removing $u_1$ (which is matched in $\vecmv[u]$) from $\vecmv[u]$.
	By induction hypothesis, the symmetric difference between $M(\vecmv[u])$ and $M(\vect{y'})$ is an alternating path $(u_1,\ldots,u_l)$ such that
	(a) for all odd $i<l$, we have $(u_i,u_{i+1})\in M(\vecmv[u])$; for all even $i<l$, we have $(u_i,u_{i+1})\in M(\vect{y'})$;
	(b) from $M(\vecmv[u])$ to $M(\vect{y'})$, vertices $\{u_2,u_4,\ldots\}$ get worse, vertices $\{u_3,u_5,\ldots\}$ get better.
	
	Hence the symmetric difference between $M(\vecy)$ and $M(\vecmv[u])$ is the alternating path $(u,u_1,\ldots,u_l)$ (recall that $M(\vecy) = M(\vec{y'})\cup \{ (u,u_1) \}$).
	Statement (a) holds, and statement (b) holds for vertices $\{u_2,\ldots,u_l\}$.
	
	Now consider vertex $u_1$, which is matched to $u$ in $M(\vecy)$, and matched to $u_2$ in $M(\vecmv[u])$.
	
	If $u_1$ is passively matched (by $u$) in $M(\vecy)$, then we know that $u$ has an earlier deadline than $u_1$.
	Hence in $M(\vecmv[u])$, either $u_1$ is active, or passively matched by some $u_2$ with a deadline later than $u$.
	In other words, $u_1$ gets worse from $M(\vecy)$ to $M(\vecmv[u])$.
	
	If $u_1$ matches $u$ actively in $M(\vecy)$, then we know that $u_1$ has an earlier deadline than $u$.
	Hence when $u_1$ is considered in $\vecmv[u]$, the set of unmatched vertices (except for $u$) is identical as in $M(\vecy)$.
	Consequently, $u_1$ actively matches some vertex $u_2$ with $1-g(y_{u_2})-f(x_{u_2}^{(u_1)}) \le 1-g(y_u)-f(x_u^{(u_1)})$ (otherwise $u_1$ will not match $u$ in $M(\vecy)$).
	In other words, $u_1$ gets worse from $M(\vecy)$ to $M(\vecmv[u])$.	
\end{proof}

Next, we define the an important set of concepts called marginal ranks.

\begin{definition}[Marginal Rank]
    For any instance $G$, any vertex $u$, and any ranks $\yminus{u}$ of other vertices, the marginal rank of $u$ w.r.t.\ $G$ and $\yminus{u}$, denoted by $\marginal_u(G, \yminus{u})$, is the largest rank of $u$ such that it is passive, i.e., $\marginal_u(G, \yminus{u}) = \sup \big\{ y_u : \text{$u$ is passive in $M(y_u, \yminus{u})$} \big\}$.
\end{definition}

For any pair of neighbors $(u,v)$ where $u$'s deadline is earlier than $v$'s, we focus on the instance up to the deadline of $u$. For simplicity, we assume $u$'s deadline to be the end of the instance and define the following marginal ranks with respect to the instance right after $u$'s deadline.

\begin{definition}[Marginal Ranks $\tau$ and $\gamma$]
    Fix any instance $G$, any edge $(u, v)$, and any ranks $\yminus{uv}$ of the vertices other than $u$ and $v$.
    Let $\tau = \marginal_u(\Gminus{v}, \yminus{uv})$ be the marginal rank of $u$ w.r.t.\ instance $\Gminus{v}$ with $v$ removed, and ranks $\yminus{uv}$.
    Similarly, let $\gamma = \marginal_v(\Gminus{u}, \yminus{uv})$.
\end{definition}

\begin{definition}[Marginal Rank $\theta$]
    Fix any instance $G$, any edge $(u, v)$ in which $u$ has an earlier deadline, any rank $y_u$ of $u$, and any ranks $\yminus{uv}$ of the vertices other than $u$ and $v$.
    Let $\theta(y_u) = \marginal_v(G, (y_u, \yminus{uv}))$ be the marginal rank of $v$ w.r.t.\ $G$ and ranks $(y_u, \yminus{uv})$.
\end{definition}

In fact we are only interested in $\theta(y_u)$ for $y_u > \tau$.
The next lemma states that it suffices to consider a single value $\theta$.

\begin{lemma} \label{lem:theta}
    There exists $\theta \ge \gamma$ such that $\theta(y_u) = \theta$ for any $y_u > \tau$.
\end{lemma}

\begin{proof}
Consider the graph with $v$ removed and $y_u=\tau^+$. By the definition of $\tau$, $u$ remains unmatched before its deadline. Consider inserting $v$ with $y_v \in (\gamma,1)$. According to the definition of $\gamma$, $v$ must also be unmatched before $u$'s deadline. That is, for any $y_u \in (\tau, 1)$ and $y_v \in (\gamma,1)$ , both $u,v$ are unmatched before $u$'s deadline. Note that at this moment, the rank of $u$ does not play any role for its decision. Hence, there exists a common $\theta$ such that $v$ would matched by $u$ iff $y_v < \theta$.
\end{proof}

We remark that $\theta$ may be $1$, in which case $v$ is passive regardless of its rank $y_v$.
We will treat as a degenerate case and will handle it separately in the analysis (see Lemma~\ref{lem:bal_ranking_structure} and Figure~\ref{fig:bip2}).

The marginal ranks $\tau$, $\gamma$, and $\theta$ provide a characterization of the matching results of $u$ and $v$ as their ranks change.
This is summarized in the following lemma, whose counterpart for \ranking were shown as Lemma 4.1, 4.2, and 4.3 in \citet{soda/HuangPTTWZ19}.
See also Figure~\ref{fig:balandced_ranking} for a more visualized illustration.

\begin{figure}[t]
	\centering
	\begin{subfigure}{.45\textwidth}
		\centering
		\includegraphics[width=0.7\linewidth]{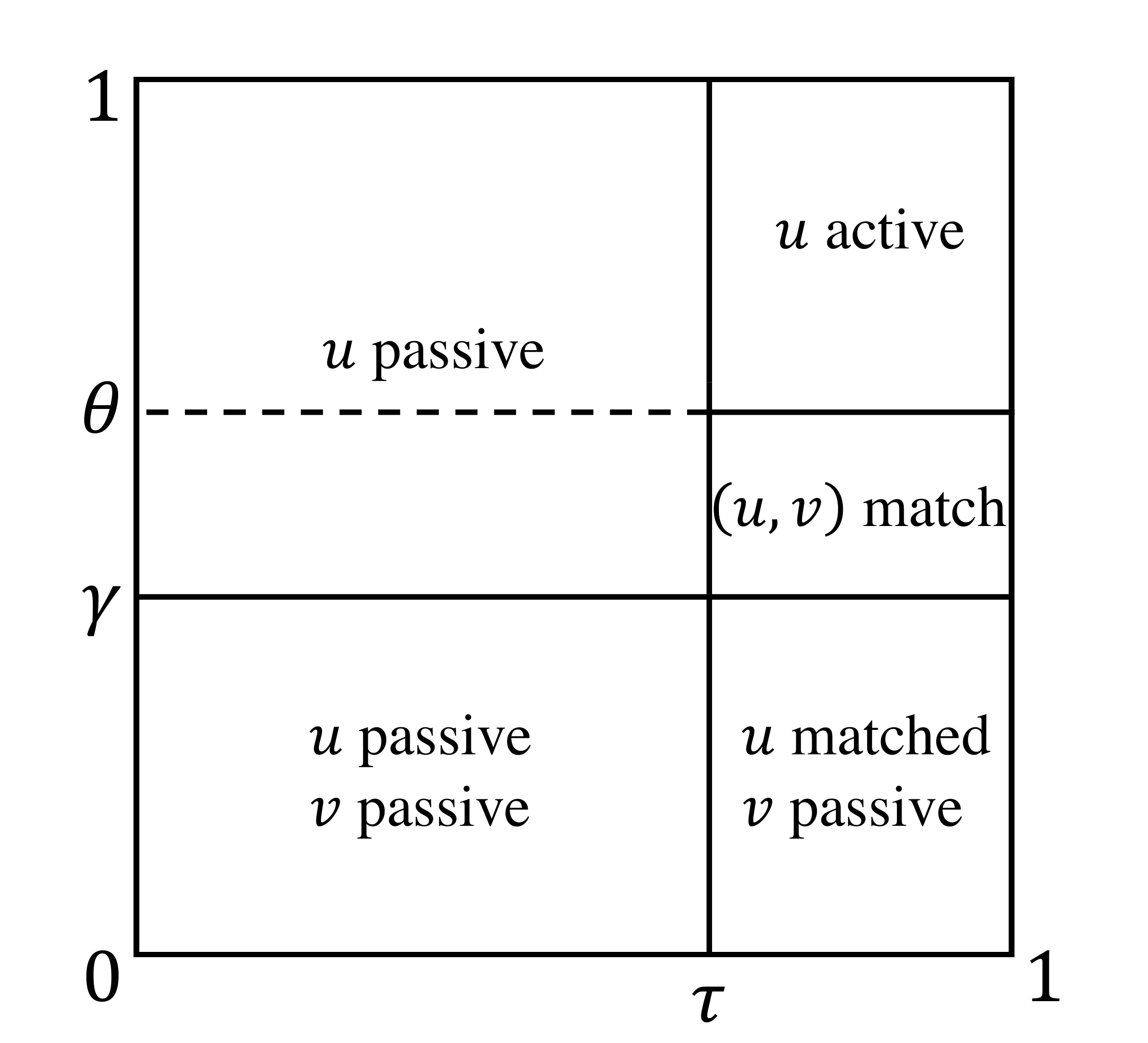}
		\vspace*{-5pt}
		\caption{General case when $\theta < 1$.}
		\label{fig:bip1}
	\end{subfigure}
	\begin{subfigure}{.45\textwidth}
		\centering
		\includegraphics[width=0.7\linewidth]{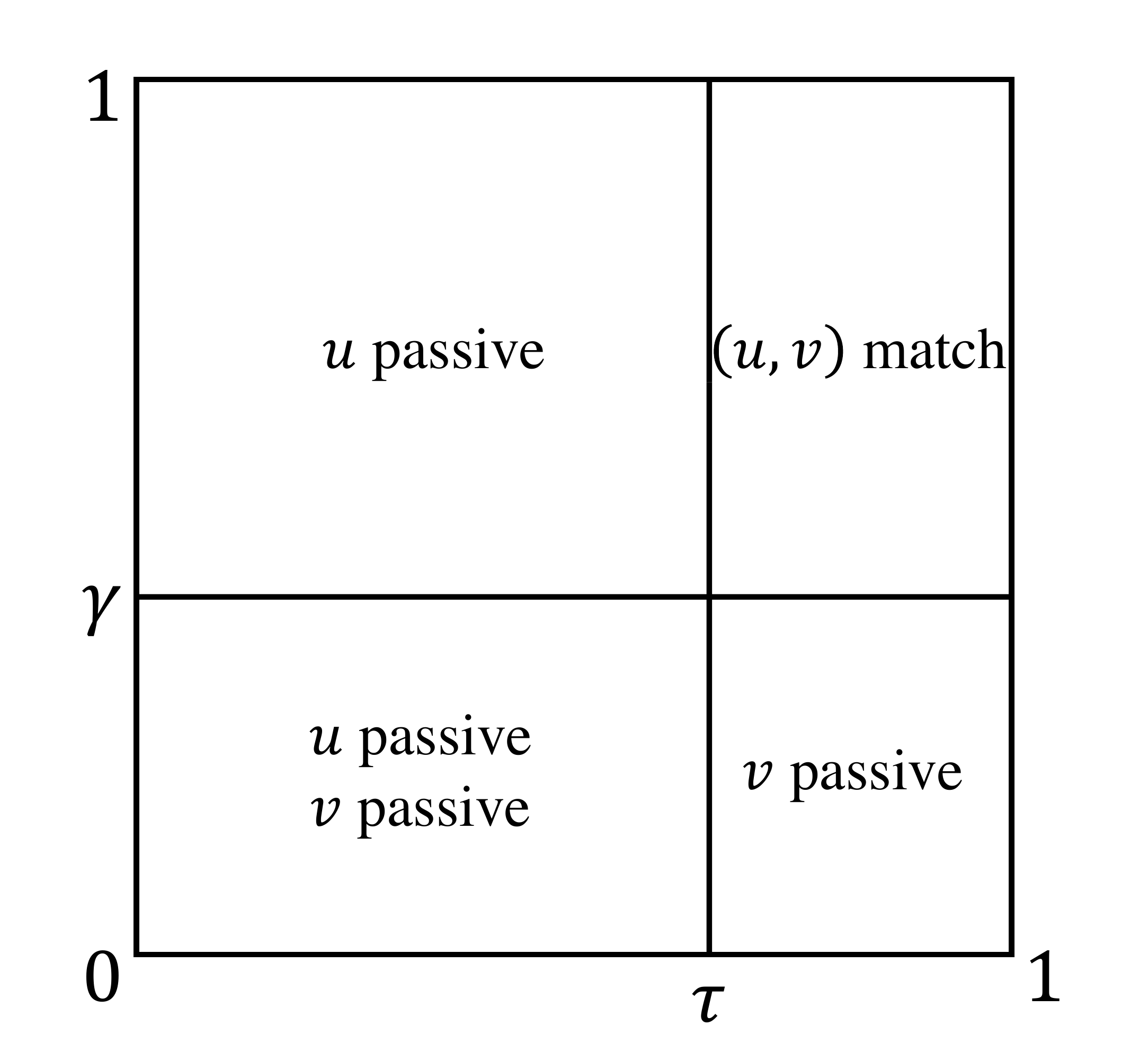}
		\vspace*{-5pt}
		\caption{Degenerate case when $\theta = 1$.}
		\label{fig:bip2}
	\end{subfigure}
	\caption{The horizontal and vertical axes correspond to $y_u,y_v$ respectively.}
	\label{fig:balandced_ranking}
\end{figure}


\begin{lemma}
    \label{lem:bal_ranking_structure}
    For any instance $G$, any edge $(u, v)$ where $u$ has an earlier deadline than $v$, any ranks $\vec{y}_{\text{-}uv}$ of other vertices, and the corresponding marginal ranks $\tau$, $\gamma$, and $\theta$, we have:
    \begin{compactitem}
        \item $u$ is passive when $y_u \in (0,\tau)$ and $y_v \in(0,1)$;
        \item $v$ is passive when $y_v \in (0,\gamma)$ and $y_u\in (0,1)$;
        \item for any $y_u \in (\tau,1)$, $v$ is matched if and only if $y_v \in (0,\theta)$;
        \item for any $y_u \in (\tau,1)$ and $y_v \in (\gamma, \theta)$, $u$ actively matches $v$;
        \item for any $y_u \in (\tau,1)$ and $y_v \in (\theta,1)$, $\alpha_u \ge 1-g(\theta)-f(x_v^{(u)})$, i.e., $u$'s gain is at least what $v$ offers at its marginal rank $y_v = \theta$;
        \item when $\theta < 1$, $u$ is matched when $y_u \in (\tau,1)$ and $y_v \in (0,\gamma)$; if we further have $u$ is active, then $\alpha_u \ge 1-g(\theta)-f(x_v^{(u)})$.
    \end{compactitem}
\end{lemma}

\begin{proof}
We prove the statements sequentially. By the definition of $\tau$, $u$ is passively matched when $y_u \le \tau$ and $v$ is removed from the graph. By Lemma~\ref{lem:alternating-path}, inserting $v$ (with any rank) to the graph cannot make $u$ worse. Hence, $u$ must be passive. Similarly, $v$ is passive when $y_v \le \gamma$. This finishes the proof of the first and the second statements.

The third and the fourth statements hold by the definition of $\theta$. Furthermore, consider when $y_u \in (\tau, 1)$ and $y_v=\theta^+$, $u$ has $v$ as a candidate but decides to choose another vertex $z$. Note that $v$ offers $g(\theta)+f(x_v^{(u)})$. We have $\alpha_u = 1-g(y_z)-f(x_z^{(u)}) \ge 1-g(\theta)-f(x_v^{(u)})$. When we further increase the rank $y_v$, $u$'s matching status shall not change. This concludes the fifth statement.

Finally, when $\theta < 1$, consider the graph with $v$ removed and when $y_u \in (\tau,1)$. This is equivalent to the case when $y_v > \theta$ and according to the previous discussion, $u$ matches a vertex $z$ and $\alpha_u \ge 1-g(\theta)-f(x_v^{(u)})$. By Lemma~\ref{lem:alternating-path}, after inserting vertex $v$ with rank $y_v \in (0,\gamma)$, $u$'s matching status becomes no worse than actively choosing $z$. In other words, $u$ must be matched when $y_u \in (\tau, 1)$ and $y_v \in (0,\gamma)$. Furthermore, if $u$ is active, $\alpha_u$ does not decrease, i.e. $\alpha_u \ge 1-g(y_z)-f(x_z^{(u)}) \ge 1- g(\theta)-f(x_v^{(u)})$.   
\end{proof}

\subsection{Analysis of \br}
\label{subsec:analysis-balanced-ranking}

Recall the randomized online primal dual framework as in Lemma~\ref{lem:primal-dual}.
Further recall that reverse weak duality in expectation holds trivially with equality by our definition of the dual variables.
It remains to show approximate dual feasibility in expectation, i.e., to lower bound $\E \big[ \alpha_u + \alpha_v \big]$.
Since the dual variables depend on functions $f$ and $g$, the lower bound will also be expressed in terms of these functions.
It shall not be surprising that the contribution from $g$ is identical to the bound by \citet{soda/HuangPTTWZ19}.
After all, the algorithm degenerates to \ranking if we let $f(x) \equiv 0$.
For brevity, we denote the lower bound by \citet{soda/HuangPTTWZ19} as a function $G:[0,1]^3 \to [0,1]$ as
\begin{equation}
    \label{eqn:G}
    G(\tau, \gamma, \theta) \eqdef
    \begin{cases}
        \displaystyle
        \int_0^{\tau} g(y_u) \dd y_u + \int_0^{\gamma} g(y_v) \dd y_v + \big(1-\tau \big)\cdot \big(1-\gamma-(1-\theta)g(\theta) \big) \\[1ex]
        \qquad\qquad\quad\displaystyle
        +~ \gamma \cdot \int_{\tau}^1 \min \big\{ (1-g(\theta)), g(y_u) \big\} \dd y_u, & \theta < 1 ~; \\[2ex]
        \displaystyle
        \int_0^{\tau} g(y_u) \dd y_u + \int_0^{\gamma} g(y_v) \dd y_v + (1-\tau) \cdot (1-\gamma) & \theta = 1 ~.
    \end{cases}
\end{equation}


We lower bound the approximate dual feasibility in the following main technical lemma.

Observe that the bound $G(\tau, \gamma, \theta)$ in Eqn.~\eqref{eqn:G} is \emph{local}, in the sense that it is achieved by taking expectation over $y_u$ and $y_v$ only, for an arbitrarily fixed $\yminus{u,v}$.
In contrast, our lower bound in Eqn.~\eqref{eqn:bal_ranking_gain} is \emph{global}, in the sense that we need to take expectation of $G(\tau,\gamma,\theta)$ over $\yminus{u,v}$. 
Additionally the bound due to function $f$ is also global, as it takes as input the lookahead water levels.

\begin{lemma}
\label{lem:bal_ranking_gain}
For any edge $(u, v)$ in which $u$ has an earlier deadline, we have:
\begin{equation} \label{eqn:bal_ranking_gain}
\E \big[\alpha_u + \alpha_v\big] \ge \E \big[ G(\tau,\gamma,\theta) \big] + F(x_u, x_v^{(u)}) ~,
\end{equation}
where $F$ is defined as:
\[
    F(x_u, x_v^{(u)}) \eqdef \int_0^{x_u} f(x) \dd x + \int_0^{x_v^{(u)}} f(x) \dd x - (1-x_u) \cdot f(x_v^{(u)})
    ~.
\]
Recall that $x_u = \Pr \big[ \text{$u$ passive} \big]$ and $x_v^{(u)} = \Pr \big[\text{$v$ passive after $u$'s deadline} \big]$.
\end{lemma}

\begin{proof}
    We first fix arbitrary ranks $\yminus{uv}$ of all vertices but $u,v$ and define $\tau, \gamma, \theta$ w.r.t. $\yminus{uv}$. We prove that 
    \begin{equation}
    \label{eqn:fix_y-uv}
    \underset{y_u,y_v}{\E}[\alpha_u+\alpha_v] \ge G(\tau, \gamma, \theta) + \underset{y_u,y_v}{\E} \Big[ \mathbbm{1}[u \text{ passive}] \cdot f(x_u') + \mathbbm{1}[v \text{ passive}] \cdot f(x_v') - \mathbbm{1}[u \text{ active}] \cdot f(x_v^{(u)}) \Big].
    \end{equation}
    
    Remark that we use $x_u' = x_u^{(w)}$ to denote $u$'s water level right after it is passively matched to $w$ and similarly $x_v' = x_v^{(w')}$ to denote $v$'s water level right after it is passively matched to $w'$. Noticed that $x_u', x_v'$ depend on the ranks $\vec{y}$. 
    
    Then, we consider the following two cases depending on whether $\theta = 1$.
    \paragraph{Case 1: $\theta < 1$.}
    We first study the non degenerate case. Referring to Figure~\ref{fig:bip1}, we have
    \begin{align*}
    & \E_{y_u,y_v}[\alpha_u+\alpha_v] \ge \int_0^1 \int_0^1 \Big( \mathbbm{1}[u \text{ passive}] \cdot f(x_u') + \mathbbm{1}[v \text{ passive}] \cdot f(x_v') \Big) \dd y_u \dd y_v \\
    & + \int_0^{\tau} g(y_u) \dd y_u + \int_0^{\gamma} g(y_v) \dd y_v + (1-\tau) \Big( (\theta-\gamma) (1-f(p_v))+(1-\theta)(1-g(\theta)-f(x_v^{(u)})) \Big) \\
    & + \int_0^{\gamma} \int_{\tau}^{1} \Big( \mathbbm{1}[u \text{ passive}] \cdot g(y_u) + \mathbbm{1}[u \text{ active}] \cdot (1-g(\theta)-f(x_v^{(u)})) \Big) \dd y_u \dd y_v.
    \end{align*}
    \begin{itemize}
    	\item The terms in the first line corresponds to all $f(\cdot)$ terms when $u,v$ are passively matched. 
    	\item The terms $\int_0^\tau g(y_u) \dd y_u$ and $\int_0^\gamma g(y_v) \dd y_v$ correspond to the gain of $\alpha_u$ and $\alpha_v$ when $y_u < \tau$ and $y_v < \gamma$ respectively. By the first and the second statements of Lemma~\ref{lem:bal_ranking_structure}, $u$ is passive when $y_u < \tau$ and $v$ is passive when $y_v < \gamma$. We only write the $g(\cdot)$ terms since we have counted the $f(\cdot)$ terms in the first line of the equation.
    	\item The term $(1-\tau)\cdot(\theta-\tau) \cdot (1-f(x_v^{(u)}))$ corresponds to the gain of $\alpha_u+\alpha_v$ when $y_u \in (\tau,1)$ and $y_v \in (\gamma,\theta)$. Note that $u,v$ matches each other in this region by the fourth statement of Lemma~\ref{lem:bal_ranking_structure}. However, we subtract $f(x_v^{(u)})$ from the gain since we have counted it in the first line of the equation.
    	\item The term $(1-\tau)\cdot (1-\theta) \cdot (1-g(\theta)-f(x_v^{(u)}))$ corresponds to the gain of $\alpha_u$ when $y_u \in (\tau, 1)$ and $y_v \in (\theta,1)$ by the fifth statement of Lemma~\ref{lem:bal_ranking_structure}.
    	\item The term in the last line corresponds to the gain of $\alpha_u$ when $y_u \in (\tau,1)$ and $y_v \in(0,\gamma)$. By the last statement of Lemma~\ref{lem:bal_ranking_structure}, $u$ is either passive ($\alpha_u = g(y_u)+f(x_u')$) or active ($\alpha_u \ge 1-g(\theta)-f(x_v^{(u)})$). When $u$ is passive, we subtract the $f(x_u')$ term since we have counted it in the first line of the equation.
    \end{itemize}
    
    Observe that $(-f(x_v^{(u)}))$ appears only when $u$ is active. Therefore,
    \begin{align*}
    \E_{y_u,y_v}[\alpha_u + \alpha_v] \ge & \int_0^1 \int_0^1 \Big( \mathbbm{1}[u \text{ passive}] \cdot f(x_u') + \mathbbm{1}[v \text{ passive}] \cdot f(x_v') - \mathbbm{1}[u \text{ active}] \cdot f(x_v^{(u)}) \Big) \dd y_u \dd y_v\\
    & + \int_0^{\tau} g(y_u) \dd y_u + \int_0^{\gamma} g(y_v) \dd y_v + (1-\tau)\cdot(\theta-\gamma+(1-\theta)(1-g(\theta))) \\
    & + \int_0^{\gamma} \int_{\tau}^{1} \Big( \mathbbm{1}[u \text{ passive}] \cdot g(y_u) + \mathbbm{1}[u \text{ active}] \cdot (1-g(\theta)) \Big) \dd y_u \dd y_v \\
    \ge & \int_0^1 \int_0^1 \Big( \mathbbm{1}[u \text{ passive}] \cdot f(x_u') + \mathbbm{1}[v \text{ passive}] \cdot f(x_v') - \mathbbm{1}[u \text{ active}] \cdot f(x_v^{(u)}) \Big) \dd y_u \dd y_v\\
    & + \int_0^{\tau} g(y_u) \dd y_u + \int_0^{\gamma} g(y_v) \dd y_v + (1-\tau)\cdot(\theta-\gamma+(1-\theta)(1-g(\theta))) \\
    & + \int_0^{\gamma} \int_{\tau}^{1} \min\{ g(y_u), 1-g(\theta)\} \dd y_u \dd y_v \\
    = & \underset{y_u,y_v}{\E} \Big[ \mathbbm{1}[u \text{ passive}] \cdot f(x_u') + \mathbbm{1}[v \text{ passive}] \cdot f(x_v') - \mathbbm{1}[u \text{ active}] \cdot f(x_v^{(u)}) \Big] \\
    & + G(\tau, \gamma, \theta).
    \end{align*}
    
    \paragraph{Case 2: $\theta = 1$.}
    The only difference between the two cases is that we no longer have the gain of $\alpha_u$ when $y_u \in (\tau,)]$ and $y_v \in (0,\gamma)$. Referring to Figure~\ref{fig:bip2}, we have
    \begin{align*}
    \E_{y_u,y_v}[\alpha_u+\alpha_v] \ge & \int_0^1 \int_0^1 \Big( \mathbbm{1}[u \text{ passive}] \cdot f(x_u') + \mathbbm{1}[v \text{ passive}] \cdot f(x_v') \Big) \dd y_u \dd y_v \\
    & + \int_0^{\tau} g(y_u) \dd y_u + \int_0^{\gamma} g(y_v) \dd y_v + (1-\tau)\cdot (1-\gamma) \cdot (1-f(x_v^{(u)})) \\
    \ge & \int_0^1 \int_0^1 \Big( \mathbbm{1}[u \text{ passive}] \cdot f(x_u') + \mathbbm{1}[v \text{ passive}] \cdot f(x_v') - \mathbbm{1}[u \text{ active}] \cdot f(x_v^{(u)}) \Big) \dd y_u \dd y_v \\
    & + \int_0^{\tau} g(y_u) \dd y_u + \int_0^{\gamma} g(y_v) \dd y_v + (1-\tau)\cdot (1-\gamma)\\
    = & \underset{y_u,y_v}{\E} \Big[ \mathbbm{1}[u \text{ passive}] \cdot f(x_u') + \mathbbm{1}[v \text{ passive}] \cdot f(x_v') - \mathbbm{1}[u \text{ active}] \cdot f(x_v^{(u)}) \Big] \\
    & +G(\tau, \gamma, \theta).
    \end{align*}
    
    Next, taking expectations over the ranks $\vec{y}_{\text{-}uv}$, we have
    \begin{align}
    & \underset{\vec{y}_{\text{-}uv}}{\E}\left[\underset{y_u,y_v}{\E} \Big[ \mathbbm{1}[u \text{ passive}] \cdot f(x_u') + \mathbbm{1}[v \text{ passive}] \cdot f(x_v') - \mathbbm{1}[u \text{ active}] \cdot f(x_v^{(u)}) \Big] \right] \notag\\
    = & \underset{\vecy}{\E} \Big[ \mathbbm{1}[u \text{ passive}] \cdot f(x_u') \Big] + \underset{\vecy}{\E} \Big[ \mathbbm{1}[v \text{ passive}] \cdot f(x_v') \Big] - \underset{\vec{y}}{\E}  \Big[ \mathbbm{1}[v \text{ active}] \cdot f(x_v^{(u)}) \Big] \notag\\
    \ge & \underset{\vecy}{\E} \Big[ \mathbbm{1}[u \text{ passive}] \cdot f(x_u') \Big] + \underset{\vecy}{\E} \Big[ \mathbbm{1}[u \text{ passive}] \cdot f(x_v') \Big] - (1-x_u) \cdot f(x_v^{(u)}). \label{eqn:f-bound}
    \end{align}
    Let $v_1,\dots, v_k$ be all neighbors of $u$ whose deadlines are before $u$'s and let them be enumerated according to the order of deadlines. 
    Recall that $x_u^{(v_i)}$ is the water level of $u$ after $v_i$'s deadline. By definition, we have $x_u^{(v_i)} = \Pr[u \text{ passive at } v_i]$.
    Moreover $x_u^{(v_i)} - x_u^{(v_{i-1})} = \Pr[u \text{ is matched by } v_i]$. (For notation simplicity, let $x_u^{(v_0)}=0$.)
    Thus,
    \begin{align*}
    \underset{\vecy}{\E} \Big[ \mathbbm{1}[u \text{ passive}] \cdot f(x_u') \Big] = & \sum_{i=1}^{k} \Pr[u \text{ is mathced by } v_i] \cdot f(x_u^{(v_i)}) \\
    = & \sum_{i=1}^{k} (x_u^{(v_i)} - x_u^{(v_{i-1})}) \cdot f(x_u^{(v_i)}) \ge \int_0^{x_u} f(x_u') \dd x_u',
    \end{align*}
    where the inequality comes from the monotonicity of $f$. Similarly, 
    \[
    \underset{\vecy}{\E} \Big[ \mathbbm{1}[v \text{ passive}] \cdot f(x_v') \Big] \ge \int_0^{x_v^{(u)}} f(x_v') \dd x_v'.
    \]
    We conclude the proof by combining Equation~\eqref{eqn:fix_y-uv} and \eqref{eqn:f-bound}:
    \begin{align*}
    \underset{\vec{y}}{\E}[\alpha_u + \alpha_v] \ge & \underset{\vec{y}_{\text{-}uv}}{\E}\left[ G(\tau,\gamma,\theta) + \underset{y_u,y_v}{\E} \Big[ \mathbbm{1}[u \text{ passive}] \cdot f(x_u') + \mathbbm{1}[v \text{ passive}] \cdot f(x_v') - \mathbbm{1}[u \text{ active}] \cdot f(x_v^{(u)}) \Big] \right] \\
    \ge & \underset{\vec{y}_{\text{-}uv}}{\E} [ G(\tau,\gamma,\theta)] + \int_0^{x_u} f(x) \dd x + \int_0^{x_v^{(u)}} f(x) \dd x - (1-x_u) \cdot f(x_v^{(u)}).
    \end{align*}
\end{proof}

\paragraph{Failed Attempt: Handling $f$ and $g$ Separately.}
It remains to design functions $f$ and $g$ so that the RHS of Eqn.~\eqref{eqn:bal_ranking_gain} is at least the competitive ratio $\Gamma$.
Suppose we do not have any control for the marginal ranks $\tau$, $\gamma$, $\theta$ and water levels $x_u$, $x_v^{(u)}$, i.e., they can take any arbitrary combination of values in $[0,1]$.
Then, the designs of $f$ and $g$ become two separate problems.
\citet{soda/HuangPTTWZ19} found the optimal $g$ such that $\min_{\tau, \gamma, \theta} \{G(\tau,\gamma,\theta)\} = \Omega \approx 0.567$ to show that \ranking is $\Omega$-competitive.
Unfortunately, the bound $F(x_u, x_v^{(u)})$ for any nondecreasing function $f$ is at most $0$ at $x_u = 0$.

In order to beat the $\Omega$ competitive ratio, which is proved tight for \ranking~\cite{soda/HuangPTTWZ19}, it is crucial to establish a connection between the marginal ranks and the water levels.

\paragraph{Binding Marginal Ranks and Waterlevels.}
Fortunately, the marginal ranks threshold ranks $\tau$, $\gamma$, $\theta$ and water levels $x_u$, $x_v^{(u)}$ are \emph{not} arbitrary.
Recall from the first conclusion of Lemma~\ref{lem:bal_ranking_structure} that $u$ is passive for all $y_u\in [\tau,1]$ and $y_v\in[0,1]$.
Hence conditioned on any $\yminus{u}$, the probability that $u$ is passive is at least $\tau$.
Taking the expectation over $\yminus{u}$ yields the following lemma.

\begin{lemma}   \label{lem:binding}
    For any edge $(u, v)$ in which $u$ has an earlier deadline, we have $\E \big[ \tau \big] \le x_u$.
\end{lemma}

\paragraph{Our Final Plan.}
To utilize the above relation between $\tau$ and $x_u$, we introduce an auxiliary \emph{convex} function $\ell : [0,1] \to [0,1]$ such that $\ell(\tau)$ lower bounds $\min_{\gamma \le \theta} \{G(\tau,\gamma,\theta)\}$.
Then, we can lower bound the first term on the RHS of Eqn.~\eqref{eqn:bal_ranking_gain} as:
\[
    \E \big[ G(\tau,\gamma,\theta) \big] \ge \E \big[ \ell(\tau) \big] \ge \ell \big( \E \big[ \tau \big] \big)~. \tag{convexity of $\ell$}
\]
Further observe that $F(x_u, x_v^{(u)})$ is nondecreasing in $x_u$.
We have $F \big( x_u, x_v^{(u)} \big) \ge F \big( \E \big[ \tau \big], x_v^{(u)} \big)$ by Lemma~\ref{lem:binding}.
It remains to lower bound $\ell \big( \E \big[ \tau \big] \big) + F \big( \E \big[ \tau \big], x_v^{(u)} \big)$, for any $\E \big[ \tau \big]$ and $x_v^{(u)}$.

\medskip

A set of sufficient conditions for $\Gamma$-competitiveness w.r.t.\ functions $f$, $g$, and the competitive ratio $\Gamma$ is summarized as the next lemma. The proof of Lemma~\ref{lem:opt_f_g} is deferred to Appendix~\ref{sec:factor_revealing_lp}.

\begin{lemma}
    \label{lem:opt_f_g}
    There are increasing function $g:[0,1]\to[0,1]$, non-decreasing function $f:[0,1]\to[0,1]$ and a convex function $\ell:[0,1] \to [0,1]$ such that for $\Gamma = 0.5690$:
    \begin{align*}
        \forall \tau, \gamma \in [0, 1], \forall \theta \in [\gamma, 1) : \quad & \textstyle
\int_0^\tau g(y_u) \dd y_u + \int_0^\gamma g(y_v) \dd y_v + \big(1-\tau\big) \big( 1-\gamma -(1-\theta) g(\theta) \big)  \\
        & \textstyle
        \quad + \int_\tau^1 \gamma \cdot \min \big\{ g(y_u), 1-g(\theta) \big\} \dd y_u \ge \ell(\tau)~; \\
        \forall \tau, \gamma \in [0, 1] : \quad & \textstyle
        \int_0^\tau g(y_u) \dd y_u + \int_0^\gamma g(y_v) \dd y_v + \big(1-\tau\big) \big( 1-\gamma \big) \ge \ell(\tau)~; \\
        \forall \E[\tau], x_v^{(u)} \in [0, 1] : \quad & \textstyle
        \ell(\E[\tau]) + \int_0^{\E[\tau]} f(x) \dd x + \int_0^{x_v^{(u)}} f(x) \dd x - \big(1-\E[\tau]\big) f(x_v^{(u)}) \ge \Gamma~;\\
        \forall x \in [0,1], \forall y \in [0,x]: \quad  & f(x)-f(y) \le x-y~; \\ 
        \forall x \in [0,1], \forall y \in [0,x]: \quad  & g(x)-g(y) \ge \frac{x-y}{100}~; \\ 
        & g(1) + f(1) \le 1~.
    \end{align*}
\end{lemma}

\begin{theorem}[Theorem~\ref{thm:balanced-ranking} Restated]
    \br with the functions $f$ and $g$ chosen in Lemma~\ref{lem:opt_f_g} is $0.569$-competitive for fully online matching on bipartite graphs.
\end{theorem}

\begin{proof}
    We have discussed all the ingredients in this section.
    It remains to put them together.
    Let $f$, $g$, and $\ell$ be the functions constructed in Lemma~\ref{lem:opt_f_g}.
    Observe that $f,g$ satisfy the Lipschitzness and reverse Lipschitzness assumed in Section~\ref{subsec:predicting}. Our algorithm is well-defined. 
    Since the function $f$ and $g$ are nonnegative and $g(1) + f(1) \le 1$, the dual variables $\alpha_u$'s are nonnegative.

    Recall that $x_u = \Pr \big[ \text{$u$ passive} \big]$ and $x_v^{(u)} = \Pr \big[ \text{$v$ passive after $u$'s deadline} \big]$.
    Approximate dual feasibility in expectation follows by Lemma~\ref{lem:bal_ranking_gain} and Lemma~\ref{lem:opt_f_g} as
    \begin{align*}
        \E \big[ \alpha_u + \alpha_v \big] \ge & \E \big[ G(\tau,\gamma,\theta) \big] + F\big(x_u, x_v^{(u)}\big) 
        \ge \E \big[ \ell(\tau) \big] + F\big(x_u, x_v^{(u)}\big) \\
        \ge & \ell\big(\E[\tau]\big) + F\big(x_u, x_v^{(u)}\big) 
        \ge \ell\big(\E[\tau]\big) + F\big(\E[\tau], x_v^{(u)}\big)
        \ge \Gamma = 0.5690~.
    \end{align*}

    Finally, recall that reverse weak duality in expectation follows trivially with equality by our definition of the dual variables.
\end{proof}

\section{Eager Water-Filling Algorithm} \label{sec:ewf}

In this section we present the \ewf algorithm for fractional fully online matching and prove Theorem~\ref{thm:eager-wf}.
We first briefly summarize the competitive analysis of \wtf algorithm~\cite{soda/HuangPTTWZ19} to build intuition.
Recall that \wtf is a lazy algorithm that each vertex sits back and waits until its deadline. At the deadline of a vertex $u$, \wtf continuously matches $u$ to the unmatched neighbor with the smallest matched portion (a.k.a. water level). The algorithm simultaneously updates the dual variables. Whenever $\dd x$ fraction of edge $(u,v)$ is matched at $u$'s deadline, we increase $\alpha_u, \alpha_v$ by $(1-f(x_v))\dd x$ and $f(x_v) \dd x$ respectively, where $x_v$ is the current water level of $v$.
Huang et al.~\cite{soda/HuangPTTWZ19} conclude the competitive ratio of \wtf by showing approximate dual feasibility with an appropriate choice of $f$.

The primal-dual analysis gives an intuitive economic interpretation of the \wtf algorithm. At any moment, each vertex $v$ prices itself at $f(x_v)$ according to the current water level and offers a share of $1-f(x_v)$ to its neighbor. At the deadline of a vertex $u$, it chooses the unmatched neighbor that is willing to give $u$ the largest share of gain. From this viewpoint, however, \wtf is unnatural in the following scenario. Suppose at $u$’s arrival, it has an existing
neighbor $v$ who is willing to offer a share of the gain that is larger than what $u$ can get from being passive matched later, i.e. $1-f(x_v) \ge f(x_u)$. Why would $u$ prefer to wait as in \wtf, instead of grabbing $v$ immediately? By waiting there is risk that 1) $v$ is taken by some other
vertex before $u$’s deadline and that 2) $u$ is passively matched before its own deadline which gives a lower portion of the gain to $u$.
To this end, we propose the following variant of \wtf.

\paragraph{\ewf.}
Fix an increasing function $f:[0,1] \to [0,1]$. Initialize all $x_{uv}$'s and $\alpha_u$'s to be zero. For convenience of analysis we also fix $f(0)=0$ and $f(1) = 1$.
\begin{enumerate}
	\item Upon the arrival of a vertex $u$, $u$ continuously matches the neighbor $v$ with lowest water level if $f(x_u) + f(x_v) \leq 1$.
	The process increases $x_u$ and the lowest water level of neighbors of $u$ until $f(x_u) + f(x_v) > 1$ for all neighbor $v$ of $u$.
	
	\item At the deadline of $u$, $u$ continuously matches the neighbor $v$ with lowest water level until $x_u = 1$, or $x_v = 1$ for all neighbor $v$ of $u$.
\end{enumerate}

Note that the second step of \ewf is the same as \wtf.

In both steps, when we match $u$ with its neighbor $v$, we consider $u$ as the active vertex and $v$ as the passive vertex.
When $x_{uv}$ increases by $\dd x$, we update the dual variables $\alpha_u$ and $\alpha_v$ as follows:
\begin{equation*}
\dd\alpha_u = (1-f(x_v))\dd x \quad \text{and} \quad \dd\alpha_v = f(x_v) \dd x.
\end{equation*}


\subsection{Analysis of Eager Water-filling}

By Lemma~\ref{lem:primal-dual}, it suffices to show that for any pair of neighbors $u$ and $v$ we have $\alpha_u+ \alpha_v \geq \Gamma$ in order to prove \ewf is $\Gamma$-competitive. Unlike Balanced Ranking/Ranking, \ewf is a deterministic algorithm and thus, no randomness is involved for the dual variables $\alpha_u$'s.
Fix any pair of neighbors $u$ and $v$, and assume $u$ has an earlier deadline than $v$.

Let $p_u$ be the water level of $u$ right \emph{before} $u$'s deadline. 
Let $p_v$ be the water level of $v$ right \emph{after} $u$'s deadline.
Let $t_u, t_v$ be the water levels of $u$ and $v$ right \emph{after} their arrivals, respectively.
We prove the following lower bound on the gain of $u$ and $v$.

\begin{lemma}\label{lemma:gain-of-u-and-v-ewf}
Right after $u$'s deadline, we have
\begin{equation}
\alpha_v + \alpha_u\ge t_v \cdot f(t_v) + \int_{t_v}^{p_v} f(x) \dd{x} + t_u \cdot f(t_u) + \int_{t_u}^{p_u} f(x) \dd{x}+ (1-p_u) \cdot (1-f(p_v)).\label{eq:gain-ewf}
\end{equation}
\end{lemma}
\begin{proof}
	Recall that when $v$ arrives, $v$ matches some neighbor actively until $x_v = t_v$.
	Moreover, when $x_v$ increases (actively) from $0$ to $t_v$, the neighbor $z$ it matches always satisfies $f(x_z) + f(t_v) \leq 1$.
	Thus when $x_v$ increases by $\dd x$ the gain of $v$ is $(1-f(x_z))\dd x \geq f(t_v)\dd x$.
	Hence right after $v$'s arrival we have $\alpha_v \geq t_v\cdot f(t_v)$.
	When $x_v$ further increases from $t_v$ to $p_v$ between $v$'s arrival and $u$'s deadline, $\alpha_v$ increases at the rate of $f(x_v)$.
	Thus after $u$'s deadline we have $\alpha_v \ge t_v \cdot f(t_v) + \int_{t_v}^{p_v} f(x) \dd{x}$.
	
	Similarly, right before $u$'s deadline we have $\alpha_u \ge t_u \cdot f(t_u) + \int_{t_u}^{p_u} f(x) \dd{x}$.
	If $p_v = 1$, then $(1-p_u)\cdot(1-f(p_v)) = 0$ and the statement is proved.
	Otherwise at $u$'s deadline, $x_u$ increases (actively) from $p_u$ to $1$, and $u$ always matches a neighbor with water level at most $p_v$.
	Thus after the deadline of $u$ we have $\alpha_u \ge t_u \cdot f(t_u) + \int_{t_u}^{p_u} f(x) \dd{x} + (1-p_u)\cdot (1-f(p_v))$.
	
	Putting the lower bounds of $\alpha_u$ and $\alpha_v$ together concludes the proof.	
\end{proof}

\paragraph{Comparison with \wtf.}
We make a comparison to the competitive analysis of \wtf by Huang et al.~\cite{soda/HuangPTTWZ19}. Let $p_u,p_v$ be defined in the same way as \ewf for \wtf and dual variables be also updated in the same way. Observe that \wtf is exactly the second step of our \ewf algorithm. Huang et al. proved that
\begin{equation}
\alpha_u+\alpha_v \geq \int_0^{\pw_u} f(x) \dd x + (1-{\pw_u})(1-f(p_v)) + \int_0^{p_v} f(x)\dd x.\label{eq:gain-wf1}
\end{equation}

Observe that Eqn.~\eqref{eq:gain-ewf} is at least as good as Eqn.~\eqref{eq:gain-wf1}, because $t \cdot f(t) \geq \int_0^{t} f(x) \dd x$ for all $t$.
On the other hand, we have not shown any constraint on the values of $t_u,t_v$.
In the case when $t_u = t_v = 0$, Eqn.~\eqref{eq:gain-ewf} degenerates to Eqn.~\eqref{eq:gain-wf1}.

We continue our analysis by observing that if $v$ arrives earlier than $u$ then right after $u$'s arrival we have $f(t_u) + f(x_v) > 1$, and $x_v \leq p_v$.
Thus we have the constraint that $f(t_u) + f(p_v) > 1$.
Similarly, if $u$ arrives earlier than $v$ then we have $f(t_v) + f(p_u) > 1$.

Combining the constraints on $t_u,p_u,t_v,p_v$ with Lemma~\ref{lemma:gain-of-u-and-v-ewf}, we show that there exists function $f$ such that the total gain of $u$ and $v$ combined is strictly larger than the ratio $2-\sqrt{2}$ that is proved tight for Water-filling.

\subsection{Reformulating the Lower Bound}

It remains to find an increasing function $f$ such that the minimum of RHS of Eqn.~\eqref{eq:gain-ewf}, over possible values of $t_u,t_v,p_u,p_v$, is maximized.
In this section we reformulate the lower bound and eliminate $t_u$ and $t_v$ from the lower bound.

Since $f$ is strictly increasing, it is easy to see that the RHS of Eqn.~\eqref{eq:gain-ewf} is increasing w.r.t. both $t_u$ and $t_v$. Indeed, the function $t \cdot f(t) - \int_0^t f(x) \dd x$ is monotonically increasing in $t$.
Thus the minimum is achieved when $t_u$ and $t_v$ are minimized, subject to the constraint
\begin{align*}
f(t_u) + f(p_v) & > 1 \text{ if $v$ arrives earlier than $u$, or}\\
f(t_v) + f(p_u) & > 1 \text{ if $u$ arrives earlier than $v$}. 
\end{align*}

Let $h(\cdot) = f^{-1}(\cdot)$ be the inverse function of $f$.
Note that $h$ is also an increasing function defined on $[0,1]$ such that $h(0)=0$ and $h(1) = 1$.
For any $p\in [0,1]$, we have
\begin{equation*}
\int_0^p f(x) \dd x = p\cdot f(p) - \int_0^{f(p)} h(y) \dd y.
\end{equation*}

\begin{lemma}\label{lemma:v-arrive-earlier}
	If $v$ arrives earlier than $u$, then we have
	\begin{equation*}
	\alpha_u + \alpha_v \ge \min_q\left\{ q \cdot h(q) - \int_0^q h(y) \dd{y} + 1 - q \right\}.
	\end{equation*}
\end{lemma}
\begin{proof}
	Let $q_u=f(p_u), q_v=f(p_v)$.
	Recall that if $v$ arrives earlier than $u$ then the minimum of RHS of Eqn.\eqref{eq:gain-ewf} is achieved when $t_u = f^{-1}(1-f(p_v)) = h(1-q_v)$ and $t_v = 0$:
	\begin{equation*}
	\alpha_v + \alpha_u\ge \int_{0}^{p_v} f(x) \dd{x} + t_u \cdot f(t_u) + \int_{t_u}^{p_u} f(x) \dd{x}+ (1-p_u) \cdot (1-f(p_v)).
	\end{equation*}
	
	Observe that the derivative of RHS of the above equation over $p_u$ is $f(p_u) + f(p_v) - 1 \ge 0$, which implies that the minimum is achieved when $p_u$ is minimized, i.e., $p_u = t_u$. Thus we have
	\begin{equation*}
	\alpha_v + \alpha_u\ge \int_{0}^{p_v} f(x) \dd{x} + t_u \cdot f(t_u) + (1-t_u) \cdot (1-f(p_v)).
	\end{equation*}
	
	Using $t_u = h(1-q_v)$ we have $f(t_u) = 1-q_v = 1-f(p_v)$, which implies
	\begin{equation*}
	\alpha_v + \alpha_u\ge \int_{0}^{p_v} f(x) \dd{x} + 1-f(p_v) = h(q_v)\cdot q_v - \int_0^{q_v} h(y) \dd y + 1 - q_v.
	\end{equation*}
	
	Taking minimum of the RHS over $q_v$ yields the lemma. 
\end{proof}

\begin{lemma}\label{lemma:u-arrive-earlier}
	If $u$ arrives earlier than $v$, then we have
	\begin{equation*}
	\alpha_u + \alpha_v \ge
	\min_{q_u,q_v} \Big\{ q_u\cdot h(q_u) - \int_0^{q_u}h(y) \dd y
	+ q_v\cdot h(q_v) - \int_0^{q_v}h(y) \dd y 
	+ \int_0^{1-q_u} h(y) \dd y + (1-h(q_u)) \cdot (1-q_v) \Big\}.
	\end{equation*}
\end{lemma}
\begin{proof}
	Let $q_u=f(p_u), q_v=f(p_v)$.
	If $u$ arrives earlier than $v$ then the minimum of RHS of Eqn.\eqref{eq:gain-ewf} is achieved when $t_v = f^{-1}(1-f(p_u)) = h(1-q_u)$ and $t_u = 0$:
	\begin{align*}
	& \alpha_v + \alpha_u \ge t_v \cdot f(t_v) + \int_{t_v}^{p_v} f(x) \dd{x} + \int_{0}^{p_u} f(x) \dd{x} + (1-p_u) \cdot (1-f(p_v)) \\
	= & t_v\cdot f(t_v) + \int_{0}^{p_v} f(x) \dd{x} - \int_0^{t_v} f(x)\dd x + \int_{0}^{p_u} f(x) \dd{x} + (1-p_u) \cdot (1-f(p_v)) \\
	= & h(1-q_u)\cdot (1-q_u) + \left( h(q_v)\cdot q_v - \int_0^{q_v}h(y)\dd y \right) - \left( h(1-q_u)\cdot (1-q_u) - \int_0^{1-q_u} h(y) \dd y \right) \\
	& \qquad + \left( h(q_u)\cdot q_u - \int_0^{q_u}h(y)\dd y \right) + (1-h(q_u))\cdot (1-q_v) \\
	= & q_u\cdot h(q_u) - \int_0^{q_u}h(y) \dd y
	+ q_v\cdot h(q_v) - \int_0^{q_v}h(y) \dd y 
	+ \int_0^{1-q_u} h(y) \dd y + (1-h(q_u)) \cdot (1-q_v).
	\end{align*}
	
	Taking minimum of the RHS over $q_u$ and $q_v$ yields the lemma. 
\end{proof}

Finally, we use factor revealing lp techniques to find function $h$ with the following property. The proof of Lemma~\ref{lemma:lower-bound-ratio-ewf} is deferred to Appendix~\ref{sec:factor_revealing_lp}.

\begin{lemma}\label{lemma:lower-bound-ratio-ewf}
	There exists an increasing function $h:[0,1]\mapsto [0,1]$ such that for $\Gamma = 0.5926$: 
	\begin{align}
	\forall q\in[0,1], \qquad & q \cdot h(q) - \int_0^q h(y) \dd{y} + 1 - q \geq \Gamma, \label{eqn:v_earlier}\\
	\forall q_u,q_v\in[0,1],\qquad & q_u\cdot h(q_u) - \int_0^{q_u}h(y) \dd y
	+ q_v\cdot h(q_v) - \int_0^{q_v}h(y) \dd y \notag \\ 
	& \qquad + \int_0^{1-q_u} h(y) \dd y + (1-h(q_u)) \cdot (1-q_v) \geq \Gamma, \label{eqn:u_earlier}\\
	& h(0) = 0, h(1) = 1.
	\end{align}
\end{lemma}

\begin{theorem}[Theorem~\ref{thm:eager-wf} Restated]
	Eager Water-filling with the function $f=h^{-1}$ where $h$ is chosen in Lemma~\ref{lemma:lower-bound-ratio-ewf} is $0.592$-competitive for fractional fully online matching on general graphs.
\end{theorem}
\begin{proof}
	We conclude the competitive ratio of \ewf by putting the lemmas together. Approximate dual feasibility follows by the two cases.
	If $v$ arrives earlier than $u$, we have
	\begin{align*}
	\alpha_u+\alpha_v \ge & \min_q\left\{ q \cdot h(q) - \int_0^q h(y) \dd{y} + 1 - q \right\} \tag{Lemma~\ref{lemma:v-arrive-earlier}} \\
	\ge & \Gamma = 0.592~. \tag{Eqn.~\eqref{eqn:v_earlier}}
	\end{align*}
	If $u$ arrives earlier than $v$, we have
	\begin{align*}
	\alpha_u + \alpha_v \ge & \min_{q_u,q_v} \bigg\{ q_u\cdot h(q_u) - \int_0^{q_u}h(y) \dd y + q_v\cdot h(q_v) - \int_0^{q_v}h(y) \dd y \\
	& \phantom{\min_{q_u,q_v} \bigg\{} + \int_0^{1-q_u} h(y) \dd y + (1-h(q_u)) \cdot (1-q_v) \bigg\} \tag{Lemma~\ref{lemma:u-arrive-earlier}} \\
	\ge & \Gamma = 0.592~. \tag{Eqn.~\eqref{eqn:u_earlier}}
	\end{align*}
	Finally, recall that reverse weak duality follows trivially with equality by our definition of the dual variables.
\end{proof}

\section{Future Directions}

\paragraph{Balanced Ranking vs.\ Ranking on General Graphs.}
An immediate next question about Balanced Ranking is whether it is still better than Ranking on general graphs.
This is beyond the scope of the current paper since a tight analysis of Ranking remains elusive.
An easier task is to show that Balanced Ranking is strictly better than $0.5211$-competitive on general graphs.
We leave these questions for future research.

\paragraph{Balanced Ranking with Eager Matches.}
Another interesting direction is to explore the power of eager matches in integral fully online matching algorithms.
There is a natural definition of Eager Ranking where a vertex $v$ may be eagerly matched on its arrival to a neighbor $u$ if $1-g(y_u) \ge g(y_x)$.
However, it is at best $\Omega \approx 0.567$-competitive due to the same hard instance for Ranking by \citet{stoc/HuangKTWZZ18}.
There is also a natural definition of Eager Balanced Ranking but its analysis seems to require ideas beyond those in this paper.

{
	\bibliography{matching}
	\bibliographystyle{plainnat}
}

\newpage

\appendix

\section{Approximate Solutions to the Differential Equations}
\label{sec:factor_revealing_lp}

In this section, we explain in detail how we use factor revealing LP techniques to construct functions $f,g,\ell$ in Lemma~\ref{lem:opt_f_g} and $h$ in Lemma~\ref{lemma:lower-bound-ratio-ewf}.

\subsection{Proof of Lemma~\ref{lem:opt_f_g}}
Recall that we need increasing function $g : [0, 1] \mapsto [0, 1]$, non-decreasing function $f : [0, 1] \mapsto [0, 1]$, and convex function $\ell : [0, 1] \mapsto [0, 1]$ such that:
\begin{align}
\forall \tau, \gamma \in [0, 1], \theta\in [\gamma,1) : \quad &
\int_0^\tau g(y) \dd y + \int_0^\gamma g(y) \dd y_v + \big(1-\tau\big) \big( 1-\gamma -(1-\theta) g(\theta) \big) \notag \\
& 
\quad + \gamma \cdot \int_\tau^1 \min \big\{ g(y), 1-g(\theta) \big\} \dd y \ge \ell(\tau)
\label{eqn:proof-g-constraint-1}
~; \\
\forall \tau, \gamma\in[0,1] : \quad &
\int_0^\tau g(y) \dd y + \int_0^\gamma g(y) \dd y + \big(1-\tau\big) \big( 1-\gamma \big) \ge \ell(\tau)
\label{eqn:proof-g-constraint-2}
~; \\
\forall \E[\tau], x_v^{(u)} \in [0,1]: \quad &
\ell(\E[\tau]) + \int_0^{\E[\tau]} f(x) \dd x + \int_0^{x_v^{(u)}} f(x) \dd x - \big(1-\E[\tau]\big) f(x_v^{(u)}) \ge \Gamma = 0.569
\label{eqn:proof-f-constraint}
~; \\[1ex]
\forall x \in [0,1], \forall y \in [0,x]: \quad  & f(x)-f(y) \le x-y~; \\ 
\forall x \in [0,1], \forall y \in [0,x]: \quad  & g(x)-g(y) \ge \frac{x-y}{100}~; \\ 
& g(1) + f(1) \le 1
\label{eqn:proof-boundary-constraint}
~.
\end{align}

In the following we construct functions $f,g$ and $l$.
For any positive integer $n$, let $[0, 1]_n$ denote the set of multiples of $\frac{1}{n}$ between $0$ and $1$:
\begin{equation*}
[0, 1]_n = \bigg\{ \frac{i}{n} : 0 \le i \le n \bigg\}.
\end{equation*}

Fix $0\leq f(0) \leq f(\frac{1}{n}) \leq \ldots \leq f(1) = 1$.
For each $x = \bar{x} + \frac{z_x}{n}$, where $\bar{x}\in [0,1]_n$ and $z_x\in [0,1)$, define $f(x) = (1-z_x)\cdot f(\bar{x}) + z_x\cdot f(\bar{x} + \frac{1}{n})$.
That is, function $f$ on points outside $[0, 1]_n$ is defined to be a linear interpolation of the function values on two nearest points in $[0, 1]_n$.

By the above definition, $f$ is uniquely defined by $\{f(x)\}_{x\in[0,1]_n}$.
In the following, we restrict our choice of function $f$ to be of this specific form.
Similarly, we strictly functions $g$ (resp. $l$) to be defined by $\{g(y)\}_{y\in[0,1]_n}$ (resp. $\{\ell(\tau)\}_{\tau\in[0,1]_n}$).

Note that for function $f$ defined this way and any $t \in [0,1]_n$, we have
\begin{equation*}
\int_0^t f(x) \dd x = \sum_{x\in[0,1]_n: x < t} \frac{f(x) + f(x+\frac{1}{n})}{2n}.
\end{equation*}

Similarly, we have $\int_0^{t} g(y) \dd y = \sum_{y\in [0,1]_n: y<t} \frac{g(y) + g(y+\frac{1}{n})}{2n}$ for all $t\in[0,1]_n$.

It remains to compute $\{ f(x),g(x),\ell(x) \}_{x\in[0,1]_n}$ that induce functions satisfying the above constraints.
Specifically, we have the following set of discretized linear constraints.

We formulate the following linear program $({LP_n})$, in which $\{ f(x),g(x),\ell(x) \}_{x\in[0,1]_n}$ are the variables.
The objective of ${LP_n}$ is to maximize variable $r$, subject to the following constraints.

\paragraph{Monotonicity.}
For any $x\in [0,1]_n, ~ x<1$:
\begin{align} 
f(x) & \le f(x+\frac{1}{n}) ~; \label{eqn:f-monotone}\\
g(x) & < g(x+\frac{1}{n}) ~. \label{eqn:g-monotone}
\end{align}

\paragraph{Boundary Condition.}
\begin{equation*} 
f(0)\ge 0,\quad g(0) \ge 0,\quad f(1) + g(1) \le 1~.
\end{equation*}

\paragraph{Lipschitzness.}
For any $x\in[0,1]_n,~ x< 1$:
\begin{align}
f(x+\frac{1}{n}) - f(x) & \le \frac{1}{n}~; \label{eqn:f-lipschitz} \\
g(x+\frac{1}{n}) - g(x) & \le \frac{1}{n}~. \label{eqn:g-lipschitz}
\end{align}

\paragraph{Reverse Lipschitzness.}
For any $x \in [0,1]_n,$ $x < 1$:
\begin{equation}
g(x+\frac{1}{n}) - g(x) \ge \frac{1}{100n}~. \label{eqn:g-reverse-lipschitz}
\end{equation}

\paragraph{Convexity.}
For any $x\in [0,1]_n,~ 0<x<1$:
\begin{equation}
\ell(x) \le \frac{1}{2}\cdot \left(\ell(x-\frac{1}{n}) + \ell(x+\frac{1}{n}) \right)~. \label{eqn:ell-convex}
\end{equation}

\paragraph{Strengthened Constraints.}
For any $\tau, \gamma, \theta \in [0, 1]_n$ such that $\gamma \le \theta + \frac{1}{n}\le 1$:
\begin{align}
\ell(\tau) + \frac{1}{2n^2} \le & \int_0^\tau g(y)\dd y + \int_0^\gamma g(y)\dd y
+ \big(1-\tau\big) \big( 1-\gamma -(1-\theta) g(\theta) \big) \notag \\[1ex]
+ & \gamma \cdot \sum_{y \in [0, 1]_n : \tau \le y < 1} \frac{\min \big\{ g(y), 1-g\big(\theta + \frac{1}{n}\big) \big\} + \min \big\{ g(y + \frac{1}{n}), 1-g\big(\theta + \frac{1}{n}\big) \big\}}{2n}~. \label{eqn:strengthen-g-constraint}
\end{align}

For any $\tau, \gamma \in [0,1]_n$:
\begin{equation}
\ell(\tau) + \frac{1}{4n^2} \le \int_0^\tau g(y) \dd y + \int_0^\gamma g(y) \dd y
+ \big(1-\tau\big) \big( 1-\gamma \big)~. \label{eqn:strengthen-g-constraint-2}
\end{equation}

For any $p, q \in [0, 1]_n$:
\begin{equation}
r + \frac{1}{4n^2} \le  \ell(p) + \int_0^{p} f(x)\dd x + \int_0^{q} f(x)\dd x - \big(1-p\big) f(q)~. \label{eqn:strengthen-f-constraint}
\end{equation}

The following claim is verified using the Gurobi LP solver\footnote{Our code is available at \url{https://github.com/denil1111/Fully-Online-Maching-Improved-Algorithms}.}.

\begin{claim}\label{claim:lpn-ratio}
	For $n=100$, the optimal objective of ${LP_n}$ is at least $\Gamma = 0.569$. 
\end{claim}


We are left to prove that the optimal solution for $LP_n$, where $\{ f(x),g(x),\ell(x) \}_{x\in [0,1]_n}$ are the variables, defines the desired functions $f,g$ and $\ell$.

First, observe that the monotonicity of $g$ and $f$ follows from Eqn.~\eqref{eqn:g-monotone} and~\eqref{eqn:f-monotone} and the linear interpolations.
Similarly, the convexity of $\ell$ follows from Eqn.~\eqref{eqn:ell-convex} and the linear interpolation.
The Lipschitzness of $f$ and reverse Lipschitzness of $g$ follows from Eqn.~\eqref{eqn:f-lipschitz} and \eqref{eqn:g-reverse-lipschitz} and the linear interpolations.
The boundary condition of $g(1) + f(1) \le 1$ is explicitly stated.
It remains to verify Eqn.~\eqref{eqn:proof-g-constraint-1}, \eqref{eqn:proof-g-constraint-2} and \eqref{eqn:proof-f-constraint}.

We first prove some useful tools to ease the analysis.
In the following, for any $t\in[0,1]$, we define $\bar{t}, z_t$ and $\hat{t}$ such that $t = \bar{t} + \frac{z_t}{n}$ and $\hat{t} = \min\{\bar{t} + \frac{1}{n},1\}$.
In other words, $\bar{t}$ and $\hat{t}$ are the two points in $[0,1]_n$ nearest to $t$, where we define $\bar{t} = \hat{t} = 1$ for $t=1$.
Note that we have $t = (1-z_t)\cdot \bar{t} + z_t\cdot \hat{t}$ and
\begin{equation*}
f(t) = (1-z_t)\cdot f(\bar{t}) + z_t\cdot f(\hat{t}), \quad g(t) = (1-z_t)\cdot g(\bar{t}) + z_t\cdot g(\hat{t}).
\end{equation*}

\begin{claim}\label{claim:error-of-int}
	For any $t\in [0,1]$, we have
	\begin{equation*}
	\int_0^t f(x)\dd x - \left( (1-z_t)\cdot \int_0^{\bar{t}}f(x)\dd x + z_t\cdot \int_0^{\hat{t}}f(x)\dd x \right) \in [-\frac{1}{8n^2},0].
	\end{equation*}
	The same holds for function $g$.
\end{claim}
\begin{proof}
	By the linear interpolation definition of $f$,
	\begin{align*}
	& \int_0^t f(x)\dd x - \left( (1-z_t)\cdot \int_0^{\bar{t}}f(x)\dd x + z_t\cdot \int_0^{\hat{t}}f(x)\dd x \right) = \int_{\bar{t}}^t f(x) \dd x - z_t\cdot \int_{\bar{t}}^{\hat{t}} f(x) \dd x \\
	= & \frac{z_t}{2n}\cdot \Big(f(\bar{t})+f(t)\Big) - \frac{z_t}{2n}\cdot \Big(f(\bar{t})+f(\hat{t})\Big)
	= \frac{z_t}{2n}\cdot \Big( f(t)-f(\hat{t}) \Big) \\
	= & \frac{z_t(1-z_t)}{2n}\cdot \Big( f(\bar{t})-f(\hat{t}) \Big) \in [-\frac{1}{8n^2},0], 
	\end{align*}
	where the last inequality follows by monotonicity and Lipschitzness of function $f$. 
	The proof for function $g$ is almost identical.
\end{proof}

\subsubsection{Feasibility of Eqn.~\eqref{eqn:proof-g-constraint-1}}


Recall that we need to prove for all $\tau, \gamma \in [0, 1]$ and $\theta \in [\gamma,1)$, $G(\tau,\gamma,\theta) \geq \ell(\tau)$, where
\begin{align*}
G(\tau,\gamma,\theta) \eqdef \int_0^\tau g(y) \dd y + \int_0^\gamma g(y) \dd y + \big(1-\tau\big) \big( 1-\gamma -(1-\theta) g(\theta) \big) + \gamma \int_\tau^1 \min \big\{ g(y), 1-g(\theta) \big\} \dd y.
\end{align*}

Let $\tau = \bar{\tau} + \frac{z_\tau}{n}$, where $\bar{\tau} \in [0, 1]_n$ and $z_\tau \in [0,1)$.
Let $\hat{\tau} = \min\{ \bar{\tau}+\frac{1}{n}, 1 \}$.
We define $\bar{\gamma},z_\gamma,\hat{\gamma}$ for $\gamma$, and $\bar{\theta},z_\theta,\hat{\theta}$ for $\gamma$ similarly.
By monotonicity of $g$, $G(\tau,\gamma,\theta)$ is at least
\begin{align*}
    \int_0^\tau g(y) \dd y + \int_0^\gamma g(y) \dd y + \big(1-\tau\big) \big( 1-\gamma -(1-\theta) g(\theta) \big) + \gamma \int_\tau^1 \min \big\{ g(y), 1-g\big(\hat{\theta}\big) \big\} \dd y~.
\end{align*}

We define the above equation as $G_1(\tau, \gamma, \theta)$.
Note that the only difference between $G$ and $G_1$ is that we relax $\theta$ to $\hat{\theta}$ in the last integration.

\medskip

We prove the following four inequalities, which will be building blocks of our later analysis.

For any $\tau, \gamma,\theta \in [0, 1]$, we have:
\begin{equation} \label{eqn:tau-interpolate}
    G_1(\tau, \gamma, \theta) \ge (1-z_{\tau}) \cdot G_1(\bar{\tau}, \gamma, \theta) + z_\tau \cdot G_1(\hat{\tau}, \gamma, \theta) - \frac{1}{8n^2} ~.
\end{equation}

\begin{equation}
    \label{eqn:gamma-interpolate}
    G_1(\tau, \gamma, \theta) \ge (1-z_\gamma) \cdot G_1(\tau, \bar{\gamma}, \theta) + z_\gamma \cdot G_1(\tau, \hat{\gamma}, \theta) - \frac{1}{8n^2}
    ~.
\end{equation}

\begin{equation}
    \label{eqn:theta-interpolate}
    G_1(\tau, \gamma, \theta) \ge (1-z_\theta) \cdot G_1(\tau, \gamma, \bar{\theta}) + z_\theta \cdot G_1(\tau, \gamma, \hat{\theta}) - \frac{1}{4n^2}
    ~.
\end{equation}

Further, we will show that for any $\tau, \theta \in [0, 1]_n$:

\begin{equation}
    \label{eqn:min-integration}
    \int_\tau^1 \min \big\{ g(y), 1-g\big(\hat{\theta}\big) \big\} \dd y \ge \sum_{y \in [0, 1]_n : \tau \le y < 1} \frac{\min \big\{ g(y), 1-g(\hat{\theta}) \big\} + \min \big\{ g(y + \frac{1}{n}), 1-g(\hat{\theta} ) \big\}}{2n}.
\end{equation}

Given the four inequalities, we prove that for all $\tau,\gamma\in[0,1]$ and $\theta\in[\gamma,1)$, $G(\tau,\gamma,\theta)\geq \ell(\tau)$ by:
\begin{align*}
    & G(\tau,\gamma,\theta) \ge G_1(\tau, \gamma, \theta) \\
    \ge & \big(1-z_\tau\big)
    \bigg(
        \big(1-z_\gamma\big)\big(1-z_\theta\big)\cdot G_1\big(\bar{\tau}, \bar{\gamma}, \bar{\theta}\big)
        + \big(1-z_\gamma\big)z_\theta\cdot G_1\big(\bar{\tau}, \bar{\gamma}, \hat{\theta}\big) \\
    & \phantom{\big(1-z_\tau\big) \bigg(} + z_\gamma\big(1-z_\theta\big)\cdot G_1\big(\bar{\tau}, \hat{\gamma}, \bar{\theta}\big)
        + z_\gamma z_\theta\cdot G_1\big(\bar{\tau}, \hat{\gamma}, \hat{\theta}\big)
    \bigg) \\
    & + z_\tau
    \bigg(
        \big(1-z_\gamma\big)\big(1-z_\theta\big)\cdot G_1\big(\hat{\tau}, \bar{\gamma}, \bar{\theta}\big)
        + \big(1-z_\gamma\big)z_\theta \cdot G_1\big(\hat{\tau}, \bar{\gamma}, \hat{\theta} \big) \\
    & \phantom{+z_\tau \bigg(}+ z_\gamma\big(1-z_\theta\big)\cdot G_1\big(\hat{\tau} , \hat{\gamma}, \bar{\theta}\big)
        + z_\gamma z_\theta\cdot G_1\big(\hat{\tau}, \hat{\gamma}, \hat{\theta} \big)
    \bigg) - \frac{1}{2n^2} \\
    \ge & \big(1-z_\tau\big)\cdot \ell\big(\bar{\tau}\big) + z_\tau\cdot \ell\big(\hat{\tau}\big) = \ell\big(\tau\big) \tag{by the linear interpolation of $\ell$}
    ~.
\end{align*}

Note that the second inequality follows from Eqn.~\eqref{eqn:strengthen-g-constraint} and~\eqref{eqn:min-integration}.

It remains to prove Eqn.~\eqref{eqn:tau-interpolate},~\eqref{eqn:gamma-interpolate},~\eqref{eqn:theta-interpolate} and~\eqref{eqn:min-integration}.

\paragraph{Proof of Eqn.~\eqref{eqn:tau-interpolate}.} We prove by a sequence of inequalities as follows.
\begin{align*}
& G_1(\tau, \gamma, \theta) - \bigg( (1-z_\tau) \cdot G_1(\bar{\tau}, \gamma, \theta) - z_\tau \cdot G_1(\hat{\tau}, \gamma, \theta) \bigg) \\
= & \int_0^\tau f(x)\dd x - \bigg( (1-z_\tau)\cdot \int_0^{\bar{\tau}}f(x)\dd x + z_\tau\cdot \int_0^{\hat{\tau}}f(x)\dd x \bigg)
+ \gamma\cdot \bigg( \int_{\tau}^1 \min \big\{ g(y), 1-g\big(\hat{\theta}\big) \big\} \dd y \\
& - (1-z_\tau) \cdot \int_{\bar{\tau}}^1 \min \big\{ g(y), 1-g\big(\hat{\theta}\big) \big\} \dd y - z_\tau \cdot \int_{\hat{\tau}}^1 \min \big\{ g(y), 1-g\big(\hat{\theta}\big) \big\} \dd y \bigg)\\
\ge &  -\frac{1}{8 n^2} + \gamma \Big( \int_{\tau}^{\hat{\tau}} \min \big\{ g(y), 1-g\big(\hat{\theta}\big) \big\} \dd y - (1-z_\tau) \int_{\bar{\tau}}^{\hat{\tau}} \min \big\{ g(y), 1-g\big(\hat{\theta}\big) \big\} \dd y \Big) \tag{by Claim~\ref{claim:error-of-int}} \\
\geq  & -\frac{1}{8 n^2}~. \tag{by monotonicity of $\min \big\{ g(y), 1-g\big(\hat{\theta}\big) \big\}$ w.r.t. $y$}
\end{align*}

\paragraph{Proof of Eqn.~\eqref{eqn:gamma-interpolate}.}
From Claim~\ref{claim:error-of-int}, we have the following immediately.
\begin{align*}
& G_1(\tau, \gamma, \theta) - \bigg( (1-z_\tau) \cdot G_1({\tau}, \bar{\gamma}, \theta) - z_\tau \cdot G_1({\tau}, \hat{\gamma}, \theta) \bigg) \\
= & \int_0^\gamma f(x)\dd x - \bigg( (1-z_\gamma)\cdot \int_0^{\bar{\gamma}}f(x)\dd x + z_\gamma\cdot \int_0^{\hat{\gamma}}f(x)\dd x \bigg) \geq -\frac{1}{8 n^2}~.
\end{align*}

\paragraph{Proof of Eqn.~\eqref{eqn:theta-interpolate}.}
It follows by a sequence of inequalities as follows.
\begin{align*}
& G_1(\tau,\gamma,\theta) - \bigg( (1-z_\theta) \cdot G_1(\tau, \gamma, \bar{\theta}) + z_\theta \cdot G_1(\tau, \gamma, \hat{\theta}) \bigg) \\
= & -\big(1-\tau\big) \bigg( \big(1-\theta\big) g\big({\theta}\big) - (1-z_\theta) (1-\bar{\theta})g(\bar{\theta}) - z_\theta \big(1-\hat{\theta}\big) g\big(\hat{\theta}\big) \bigg) \\
= & -(1-\tau)\cdot \frac{z_\theta(1-z_\theta)}{n} \big( g(\hat{\theta}) - g(\bar{\theta}) \big) 
\ge -\frac{1}{4n^2}~. \tag{by the Lipschitzness of $g$}
\end{align*}

\paragraph{Proof of Eqn.~\eqref{eqn:min-integration}.}
Observe that $\min \big\{ g(y), 1 - g\big(\hat{\theta}\big) \big\}$ is either linear or concave within every interval $[t, t+\frac{1}{n})$ for $t \in [0, 1]_n$.
Moreover, the latter happens only when $1 - g\big(\hat{\theta}) \in (g(t),g(t+\frac{1}{n}))$.
Let $t^*\in [0,1]_n$ be such that $1 - g\big(\hat{\theta}) \in (g(t^*),g(t^*+\frac{1}{n}))$.
If such $t^*$ does not exist, or $t^* < \tau$ then Eqn.~\eqref{eqn:min-integration} trivially holds.
Otherwise for all $\tau\in [0,1]_n$ we have:
\begin{align*}
& \int_\tau^1 \min \big\{ g(y), 1-g\big(\hat{\theta}\big) \big\} \dd y - \sum_{y \in [0, 1]_n : \tau \le y < 1} \frac{\min \big\{ g(y), 1-g(\hat{\theta}) \big\} + \min \big\{ g(y + \frac{1}{n}), 1-g(\hat{\theta} )  \big\}}{2n}. \\
= & \int_{t^*}^{t^*+\frac{1}{n}} \min \{ g(y), 1-g(\hat{\theta}) \} \dd y - \frac{1}{2n}\Big(g(t^*)+1-g(\hat{\theta})\Big) \ge 0~.
\end{align*}

\subsubsection{Feasibility of Eqn.~\eqref{eqn:proof-g-constraint-2}}
Recall that we need to prove for all $\forall \tau, \gamma\in [0,1]$: $G_2(\tau,\gamma) \geq \ell(\tau)$, where
\begin{align*}
G_2(\tau,\gamma) \eqdef \int_0^\tau g(y) \dd y + \int_0^\gamma g(y) \dd y + \big(1-\tau\big) \big( 1-\gamma \big)~.
\end{align*}

By Claim~\ref{claim:error-of-int} we have
\begin{align*}
& G_2(\tau,\gamma) \geq (1-z_\tau)\cdot G_2(\bar{\tau},\gamma) + z_\tau\cdot G_2(\hat{\tau},\gamma) - \frac{1}{8n^2} \\
\geq & (1-z_\tau)\cdot \Big( (1-z_\gamma) G_2(\bar{\tau},\bar{\gamma}) + z_\gamma G_2(\bar{\tau},\hat{\gamma}) \Big) + z_\tau\cdot \Big( (1-z_\gamma) G_2(\hat{\tau},\bar{\gamma}) + z_\gamma G_2(\hat{\tau},\hat{\gamma}) \Big) - \frac{1}{4n^2} \\
\geq & (1-z_\tau)\cdot \ell(\bar{\tau}) + z_\tau\cdot \ell(\hat{\tau}) = \ell(\tau),
\end{align*}
where the second inequality follows from Eqn.~\eqref{eqn:strengthen-g-constraint-2}.

\subsubsection{Feasibility of Eqn.~\eqref{eqn:proof-f-constraint}}
For convenience, let $p = \E[\tau]$ and $q = x_v^{(u)}$.
Recall that we need to prove for all $p,q\in [0,1]$ that $G_3(p,q)\ge \Gamma$, where
\begin{align*}
G_3(p,q)\eqdef \ell(p) + \int_0^{p} f(x) \dd x + \int_0^{q} f(x) \dd x - \big(1-p\big) f(q)~.
\end{align*}

Also recall from Eqn.~\eqref{eqn:strengthen-f-constraint} and Claim~\ref{claim:lpn-ratio} that for any $p, q \in [0, 1]_n$, we have
\begin{equation}\label{eqn:g-3-at-least-Gamma}
G_3(p,q) \geq \Gamma + \frac{1}{4n^2}~.
\end{equation}

By Claim~\ref{claim:error-of-int} and linear interpolation of $\ell$ and $f$, for all $p,q\in [0,1]$ we have
\begin{align*}
& G_3(p,q) \geq (1-z_{p})\cdot G_3(\bar{p},q) + z_{p}\cdot G_3(\hat{p},q) - \frac{1}{8n^2} \\
\geq & (1-z_{p}) \Big( (1-z_{q}) G_3(\bar{p},\bar{q}) + z_{q} G_3(\bar{p},\hat{q}) \Big) + z_{p} \Big( (1-z_{q}) G_3(\hat{p},\bar{q}) + z_{q} G_3(\hat{p},\hat{q}) \Big) - \frac{1}{4n^2} \\
\geq & (1-z_{p})\cdot \Gamma + z_{p}\cdot \Gamma = \Gamma. \tag{by Eqn.~\eqref{eqn:g-3-at-least-Gamma}.}
\end{align*}

\subsection{Proof of Lemma~\ref{lemma:lower-bound-ratio-ewf}}

Recall that we need an increasing function $h$ such that:
\begin{align}
\forall q\in[0,1], \qquad & q \cdot h(q) - \int_0^q h(y) \dd{y} + 1 - q \geq \Gamma, \label{eqn:proof_v_earlier}\\
\forall q_u,q_v\in[0,1],\qquad & q_u\cdot h(q_u) - \int_0^{q_u}h(y) \dd y
+ q_v\cdot h(q_v) - \int_0^{q_v}h(y) \dd y \notag \\ 
& \qquad + \int_0^{1-q_u} h(y) \dd y + (1-h(q_u)) \cdot (1-q_v) \geq \Gamma=0.592, \label{eqn:proof_u_earlier}\\
& h(0) = 0, h(1) = 1.
\end{align}


Fix $0 = h(0) < h(\frac{1}{n}) < \ldots < h(1) = 1$.
For each $y = \bar{y} + \frac{z_y}{n}$, where $\bar{y}\in [0,1]_n$ and $z_y\in [0,1)$, define $h(y) = (1-z_y)\cdot h(\bar{y}) + z_y\cdot h(\bar{y} + \frac{1}{n})$.
That is, function $h$ on points outside $[0, 1]_n$ is defined to be a linear interpolation of the function values on two nearest points in $[0, 1]_n$.

Note that for function $h$ defined this way and any $t \in [0,1]_n$, we have
\begin{equation*}
\int_0^t h(y) \dd y = \sum_{y\in[0,1]_n: y < t} \frac{h(y) + h(y+\frac{1}{n})}{2n}.
\end{equation*}

It remains to determine $\{ h(y) \}_{y\in[0,1]_n}$.
We claim that the optimal solution for the following LP, where $\{ h(y) \}_{y\in[0,1]_n}$ are the variables, defines a function $h$ that satisfies the constraints listed in Lemma~\ref{lemma:lower-bound-ratio-ewf}.
For convenience we define $h(1+\frac{1}{n}) = h(1) = 1$.
\begin{align}
{(LP_n)}\qquad\text{maximize}\qquad & r \nonumber \\
\text{subject to}\qquad & \textstyle r \leq q \cdot h(q) - \int_0^q h(y)\dd y + 1-q - \frac{3}{4 n^2}, \qquad \forall q\in [0,1]_n \label{constraint:v-arrive-first} \\
&\textstyle r \leq q \cdot h(q) - \int_0^q h(y) \dd y + p \cdot h(p) - \int_0^p \dd y
+ \int_0^{1-q} h(y) \dd y \nonumber\\
& \qquad \textstyle + \left(1-h(q+\frac{1}{n})\right)\cdot (1-p) -\frac{7}{4n^2},
\qquad \forall q,p\in [0,1]_n \qquad \label{constraint:u-arrive-first} \\
& \textstyle h(0) = 0,\quad h(1) = h(1+\frac{1}{n}) = 1, \nonumber\\
& \textstyle h(y) < h(y+\frac{1}{n}), \qquad \forall y\in[0,1]_n,\quad y < 1   \tag{monotonicity} \\
& \textstyle h(y) \geq h(y+\frac{1}{n}) + \frac{2}{n}, \qquad \forall y\in[0,1]_n,\quad y < 1
\tag{Lipschitzness}
\end{align}


The following claim is verified using the Gurobi LP solver. \footnote{Our code is available at \url{https://github.com/denil1111/Fully-Online-Maching-Improved-Algorithms}.} 

\begin{claim}\label{claim:lpn-ratio-ewf}
	For $n=1000$, the optimal objective of ${LP_n}$ is at least $\Gamma = 0.592$. 
\end{claim}

We are left to prove that the optimal solution for ${LP_n}$, where $\{ h(y) \}_{y\in [0,1]_n}$ are the variables, defines the desired function $h$.
Note that the monotonicity constraint of $h$ is implied by the monotonicity of $h(y)$ for $y \in [0,1]_n$ and the linear interpolation of $h$. Further, $h(0)=0$ and $h(1)=1$ are stated explicitly in the above linear program. It remains to prove that the function $h$ defined by the optimal solution of ${LP_n}$ satisfies the constraints~\eqref{eqn:proof_v_earlier} and \eqref{eqn:proof_u_earlier}.

Let $h$ be defined by the optimal solution $\{ h(y) \}_{y\in[0,1]_n}$ of ${LP_n}$ with $n=1000$.
Fix any $q \in [0,1]$.
Let $q = \bar{q}+\frac{z_q}{n}$, where $\bar{q}\in[0,1]_n$ and $z_q\in [ 0,1 )$.
Let $\hat{q} = \bar{q}+\frac{1}{n}$.
Observe that we have $q = (1-z_q)\cdot \bar{q}+z_q\cdot \hat{q}$ and $h(q) = (1-z_q)\cdot h(\bar{q}) + z_q\cdot h(\hat{q})$.	

\begin{claim}\label{claim:error-of-int-ewf}
	Let $H(q) = q \cdot h(q) - \int_0^q h(y) \dd{y}$, we have
	\begin{equation*}
	H(q) - \Big( (1-z_q)\cdot H(\bar{q})+z_q\cdot H(\hat{q}) \Big) \in [-\frac{3}{4n^2},0).
	\end{equation*}
\end{claim}
\begin{proof}
	Recall that $h$ is defined by $\{h(y)\}_{y\in[0,1]_n}$, which satisfies the constraints of ${LP_n}$.
	
	We first consider the first term $q \cdot h(q)$. Observe that
	\begin{align*}
	& q\cdot h(q) - \Big( (1-z_q)\cdot\bar{q}\cdot h(\bar{q}) + z_q\cdot\hat{q}\cdot h(\hat{q}) \Big) = \frac{z_q(1-z_q)}{n}\cdot \Big( h(\bar{q})-h(\hat{q}) \Big) \in [-\frac{1}{2n^2},0) ,
	\end{align*}
	where the last step follows from the monotonicity and Lipschitzness of $\{h(y)\}_{y\in[0,1]_n}$.
	
	Next we consider the second term $\int_0^q h(y) \dd{y}$. Observe that
	\begin{align*}
	& \int_0^q h(y)\dd y - \left( (1-z_q)\int_0^{\bar{q}}h(y)\dd y + z_q\int_0^{\hat{q}}h(y)\dd y \right)
	= \int_{\bar{q}}^q h(y)\dd y - z_q\int_{\bar{q}}^{\hat{q}}h(y)\dd y \\
	= & \frac{z_q}{2n}\cdot \Big( h(\bar{q})+h(q) \Big) - \frac{z_q}{2n}\cdot \Big( h(\bar{q})+h(\hat{q}) \Big)
	= \frac{z_q(1-z_q)}{2n}\cdot \Big( h(\bar{q}) - h(\hat{q}) \Big) \in [-\frac{1}{4n^2},0).
	\end{align*}
	
	Combining the two lower bounds concludes the proof.
\end{proof}

\subsubsection{Feasibility of Eqn.~\eqref{eqn:proof_v_earlier}}
By Eqn.~\eqref{constraint:v-arrive-first} and Claim~\ref{claim:lpn-ratio-ewf}, we have
\begin{align*}
H(q) + 1-q & \geq  (1-z_q)\cdot \Big(H(\bar{q})+1-\bar{q}-\frac{3}{4n^2} \Big) + 
z_q\cdot \Big(H(\hat{q})+1-\hat{q}-\frac{3}{4n^2} \Big) \\
& \geq (1-z_q)\cdot \Gamma + z_q\cdot \Gamma = \Gamma.
\end{align*}

\subsubsection{Feasibility of Eqn.~\eqref{eqn:proof_u_earlier}}
Let $\bar{p},\hat{p}, z_p$ (for $p$) be defined similarly as $\bar{q},\hat{q}, z_q$ (for $q$).
Note that by Claim~\ref{claim:error-of-int-ewf}, we have
\begin{align*}
& H(q) + H(p) + \int_0^{1-q} h(y) \dd y + (1-h(q)) \cdot (1-p) \\
\geq & H(q) + H(p) + \int_0^{1-q} h(y) \dd y + (1-h(\hat{q})) \cdot (1-p) \\
\geq & \Big((1-z_q)\cdot H(\bar{q})+z_q\cdot H(\hat{q}) - \frac{3}{4n^2}\Big)
+ \Big((1-z_p)\cdot H(\bar{p})+z_p\cdot H(\hat{p}) - \frac{3}{4n^2}\Big) \\
& \qquad + \Big((1-z_q)\cdot \int_0^{1-\bar{q}} h(y) \dd y + z_q\cdot \int_0^{1-\hat{q}} h(y) \dd y + \frac{z_q(1-z_q)}{2n}\cdot ( h(1-\hat{q})-h(1-\bar{q})) \Big) \\
& \qquad + \Big((1-z_q)\cdot (1-h(\hat{q}))(1-\bar{p}) + z_q\cdot (1-h(\hat{q}))(1-\hat{p}) \\
\geq & \Big((1-z_q)\cdot H(\bar{q})+z_q\cdot H(\hat{q}) - \frac{3}{4 n^2}\Big)
+ \Big((1-z_p)\cdot H(\bar{p})+z_p\cdot H(\hat{p}) - \frac{3}{4 n^2}\Big) \\
& \qquad + \Big((1-z_q)\cdot \int_0^{1-\bar{q}} h(y) \dd y + z_q\cdot \int_0^{1-\hat{q}} h(y) \dd y - \frac{1}{4n^2} \Big) \\
& \qquad + \Big((1-z_q)\cdot (1-h(\bar{q}+\frac{1}{n}))(1-\bar{p}) + z_q\cdot (1-h(\hat{q}+\frac{1}{n}))(1-\hat{p}) \\
= & (1-z_q)\cdot \Big( H(\bar{q})+H(\bar{p})+\int_0^{1-\bar{q}}h(y)\dd y + (1-h(\bar{q}+\frac{1}{n}))(1-\bar{p}) - \frac{7}{4n^2} \Big) \\
& + z_q\cdot \Big( H(\hat{q})+H(\hat{p})+\int_0^{1-\hat{q}}h(y)\dd y + (1-h(\hat{q}+\frac{1}{n}))(1-\hat{p}) - \frac{7}{4n^2} \Big) \\
\geq & (1-z_q)\cdot\Gamma + z_q\cdot\Gamma = \Gamma,
\end{align*}
where the last inequality comes from Eqn.~\eqref{constraint:u-arrive-first}.

\end{document}